\documentclass[11pt]{article}
\usepackage[hmargin=1in,vmargin=0.95in]{geometry}
\usepackage{pdfpages}
\usepackage{color}
\usepackage{fixltx2e} 
\usepackage{mparhack} 
\usepackage{epsfig}
\usepackage[english]{babel}
\usepackage{algorithm}
\usepackage[noend]{algorithmic}
\usepackage{xcolor}
\usepackage{fmtcount} 
\usepackage[cmex10]{amsmath}
\usepackage{amssymb}
\usepackage{amsthm}
\usepackage{wrapfig}
\usepackage{comment}
\usepackage{color}
\newcommand{\BO}{\mathcal{O}}

\usepackage{graphicx}
\usepackage{psfrag}

\newcommand*\eps{\varepsilon}

\newtheorem{theorem}{Theorem}[section]
\newtheorem{lemma}[theorem]{Lemma}

\newtheorem{corollary}[theorem]{Corollary}
\newtheorem{definition}[theorem]{Definition}
\newtheorem{example}[theorem]{Example}
\newtheorem{remark}[theorem]{Remark}

\newtheorem{fact}[theorem]{Fact}

\newcommand{\todo}[1]{{\begin{small}\sffamily \color{red}TODO:  #1 \end{small}}}
\newcommand{\todoI}[1]{}
\newcommand{\blueI}[1]{}

\newcommand{\tsp}{\mathop{tsp}}
\newcommand{\tspapp}{\mathop{\overline{tsp}}}
\newcommand{\dep}{\mathop{dep}}
\newcommand{\depapp}{\mathop{\overline{dep}}}

\renewcommand{\epsilon}{\varepsilon}
\newcommand{\calL}{\mathcal{L}}

\newcommand{\blue}[1]{{ \color{blue} #1 }}

\makeatletter
\newcommand*{\StartNewContent}{%
    \let\OrigLabel\label%
    \let\OrigRef\ref%
    \renewcommand{\label}[1]{\OrigLabel{FULL:##1}}%
    \renewcommand{\ref}[1]{\OrigRef{FULL:##1}}%
    \renewcommand{\label@in@display}[1]{%
        \ifx\df@label\@empty\else
            \@amsmath@err{Multiple \string\label's:
                label '\df@label' will be lost}\@eha
        \fi
        \gdef\df@label{FULL:##1}%
    }%
}
\makeatother

\newif\iffull
\newif\ifshort

\begin{document}

\author{Benjamin Dissler ~~~~~~  Stephan Holzer\thanks{Part of the work was done at ETH Zurich. At MIT the author was supported by the following grants: AFOSR Contract Number FA9550-13-1-0042, NSF Award 0939370-CCF, NSF Award CCF-1217506, NSF Award number CCF-AF-0937274.}  ~~~~~~ Roger Wattenhofer\\
~~\small disslerb@ethz.ch ~~~~~~~~~~~ \small holzer@mit.edu  ~~~~~~~~~  \small wattenhofer@ethz.ch\\
  \small ETH Zurich  ~~~~~~~~~~~~~~~~~~~~~~  MIT  ~~~~~~~~~~~~~~~~~~~~~~~  ETH Zurich\\
}

\title{
Distributed Local Multi-Aggregation and Centrality Approximation
}
\maketitle
\thispagestyle{empty}
\begin{abstract}
We study local aggregation and graph analysis in distributed environments using the message passing model. We provide a flexible framework, where each of the nodes in a set $S$--which is a subset of all nodes in the network--can perform a large range of common aggregation functions in its $k$-neighborhood. We study this problem in the CONGEST model, where in each synchronous round, every node can transmit a different (but short) message to each of its neighbors. While the $k$-neighborhoods of nodes in $S$ might overlap and aggregation could cause congestion in this model, we present an algorithm that needs time $\BO(|S|+k)$ even when each of the nodes in $S$ performs a different aggregation on its $k$-neighborhood. The framework is not restricted to aggregation-trees such that it can be used for more advanced graph analysis. We demonstrate this by providing efficient approximations of centrality measures and approximation of minimum routing cost trees.
\end{abstract}

\clearpage
\pagestyle{plain}
\setcounter{page}{1}

\newpage
\fulltrue 

\section{Introduction}
Data aggregation and analysis is one of the most basic tasks at the heart of many distributed systems and the question of aggregating and analyzing information and networks itself as efficient as possible arises daily. The result of this is a huge body of work ranging from theoretical to practical aspects focusing on optimizing e.g. speed, space, bandwidth, energy, fault-tolerance and accuracy. As already \cite{kuhn2008distributed} stated,``the database community classifies aggregation functions
into three categories: \emph{distributive} (max, min, sum, count), \emph{algebraic} (plus, minus, average, variance), and holistic (median, $k^{th}$ smallest, or $k^{th}$ largest value). Combinations of these functions are believed to support a wide range of reasonable aggregation queries.''. Often one is also interested in computing a combination of these such as, e.g., ``What is the average of the 10\% largest values?''\cite{kuhn2008distributed}.

However, most of this work considers the case in which only one node in the network aggregates information (and then often broadcasts it to all other nodes). In reality, many nodes of a large network are independently interested in aggregating data at the same time 

and restricted to their local neighborhood. That is, all nodes in a subset $S\subseteq V$ want to perform a (possibly different) aggregation, where $V$ is the set of nodes in the network.  For example 1) a few nodes in a network want to know if certain data is stored in their neighborhood, or 2) cars participating in a vehicular ad hoc network (VANET) want to aggregate information on traffic, safety warnings and parking spots from their local neighborhood, and 3) nodes who turned idle search for the busiest nodes in their respective neighborhoods to take work from them. A simple approach is to just let all nodes in $S$ perform their aggregation at the same time---however, this might lead to congestion and a worst case runtime of $\BO(|S|\cdot k)$.

In this paper we present a general framework for local multi-aggregation that allows a time- and message-optimal implementation of a wide range of aggregation functions. We do this in a general setting, where nodes in $S$ aggregate data from their $k$-neighborhood using shortest paths and achieve a runtime of $\BO(|S|+k)$ to do so.

We show how to perform aggregation and graph analysis by aggregating through all possible shortest paths between any pair of nodes (not only using one path as it is usually done in an aggregation-tree). This is a powerful tool which enables us to provide an efficient approximation of Betweenness Centrality. We also present efficient approximations of Closeness Centrality and Minimum Routing Cost Trees.

\subsection{Our Contribution}
We provide a framework for local multi-aggregation that takes care of scheduling messages sent between nodes in an efficient way. One only needs to specify how nodes process incoming messages. That is designing algorithms depending on the aggregation-function at hand, which transform received messages into new messages to be sent. Using our framework, one can aggregate information not only by using a tree, but using all possible shortest paths from a root node to any other node. This has advantages for advanced computations as we demonstrate later. Thus we show two variations of the algorithm:
\begin{enumerate}
\item only one (shortest) $r$-$t$-path is needed for each $(r,t)\in S\times N_k(r)$ to perform the aggregation.
\item all shortest $r$-$t$-paths are needed for each $(r,t)\in S\times N_k(r)$ to perform the aggregation.
\end{enumerate}
The last version is e.g. of interest for computing betweenness centrality measures, which is a measure that depends on all shortest paths, not only a single one. \footnote{In the algorithm we approximate the number of all shortest paths starting in a certain sampled set $S$ of nodes. Note that it does not seem possible to approximate centrality measures without performing these $|S|$ independent aggregations.}. 

To perform $|S|$ independent (possibly different) aggregations in the $k$-neighborhood of each of the nodes in $S$, our algorithm takes time $\BO(|S|+k)$, which we show to be optimal (see Remark \ref{rem:lowerbound}). As an example of an aggregation-function, which can be plugged into our framework, we consider computing the maximum value
\todoI{more fkt.}
stored in the $k$-neighborhood of each of the $S$ nodes. Further aggregation functions can be implemented in a similar way. Different root nodes in $S$ can perform different aggregations at the same time. Based on this, we show how to approximate centrality measures and minimum routing cost trees and obtain the following theorems:

\begin{theorem}\label{theo:DLGruntime}
Algorithm\iffull~\ref{alg:DLGcomp}\fi~\textsc{DLGcomp} computes $|S|$ valid directed leveled graphs with depth $k$, in time $\BO(|S|+k)$. And Algorithm\iffull~\ref{alg:DLGagr}\fi~\textsc{DLGagr} aggregates information through $|S|$ directed leveled graphs with depth $k$, in time $\BO(|S|+k)$.
\end{theorem}

\begin{theorem}\label{theo:tree}
The tree variations of Algorithms\iffull~\ref{alg:DLGcomp}\fi~\textsc{DLGcomp} and\iffull~\ref{alg:DLGagr}\fi~\textsc{DLGagr} compute $|S|$ trees $T_r$ in time $\BO(|S|+D_\omega)$ and aggregate information through these $|S|$ trees, in time $\BO(|S|+D_\omega)$.
\end{theorem}

\begin{theorem}\label{theo:bcapprox}
\iffull
Algorithm \ref{alg:BC_setup_controlling} computes 
\fi
an estimation of the betweenness centrality for at least $\hat{n}$ nodes. For all those nodes holds: for $\frac{1}{n^c}<\epsilon'<\frac{1}{2}$ and $c\geq0$, if the betweenness centrality of a node $v$ is $n^2/t$ for some constant $t\geq 1$, then with probability at least $ 1-2\epsilon'$ its betweenness centrality can  $(\times,\frac{1}{\epsilon'})$-approximated with computation of $|S|=t\epsilon'$ samples of source nodes. 
\iffull
Algorithm \ref{alg:BC_setup_controlling} 
\fi
 runs in time $\BO\left(|S|+D_\omega\cdot\max\left(1,\log\frac{|S|}{D_h}\right)\right)$.
\end{theorem}

\begin{theorem}\label{theo:ccapprox}
In a weighted graph $G=(V,E,\omega)$, when the weights $\omega$ denote the distance between two nodes, 
\iffull
Algorithm \ref{alg:CC} approximates 
\fi
closeness centrality of all nodes with an inverse additive error of $\epsilon D_\omega$ and runs in $\BO\left(\frac{log(n)}{\epsilon^2}+D_\omega\right)$ time.
\end{theorem}

\begin{theorem}
Algorithm\iffull~\ref{alg:MRCT}\fi~\textsc{MRCT} computes a $(\times,2)$-approximation of the S-MRCT problem in time $\BO(|S|+D_\omega)$, when $\omega(e)$ is uniform or corresponds to the time needed to transmit a message through edge $e$..
\label{theo:mrctapprox}
\end{theorem}
\begin{theorem}\label{theo:mrctapprox-rand}
Algorithm~\textsc{MRCTrand} computes a $(\times,2+\varepsilon)$-approximation of the S-MRCT problem in time $\BO(D_\omega + \log(n))$, when $\omega(e)$ is uniform or corresponds to the time needed to transmit a message through edge $e$.
\end{theorem}
Algorithm and proof of Theorem~\ref{theo:DLGruntime} is stated in Section \ref{sec:aggregate}.
Algorithm and proof of Theorem~\ref{theo:tree} is stated in Section \ref{sec:tree}.
Algorithm and proof of Theorem~\ref{theo:bcapprox} is stated in Section \ref{sec:BC}.
Algorithm and proof of Theorem~\ref{theo:ccapprox} is stated in Section \ref{sec:CC}.
Algorithm and proof of Theorem~\ref{theo:mrctapprox-rand} is stated in
\iffull
Section \ref{app:MRCT}.
\fi

\subsection{Related Work}

\paragraph{(Local multi-)aggregation}
The authors of \cite{rajagopalan8data} survey the general area of data-aggregation in wireless sensor networks.
It is know \cite{anantvalee2007survey} that local multi-aggregation is useful in generalizations of Zone-Based Intrusion Detection Systems mentioned in. These systems detect local changes of the topology (assuming the graph does not change too fast) and react to it. The authors of \cite{Jiang:2008:DBT:1463434.1463440} perform local multi-aggregation in sensor networks to detect basic changes of the local environment, such as variations of temperature or concentration of toxic gas. If the value in one node passes a threshold, it starts a local aggregation. They also provide applications to fire fighters and emergency services in case of a flood. In \cite{Jiang:2008:DBT:1463434.1463440} changes of areal objects are detected by the cooperation of sensors in the locality of the changes and later reported to a central node.

\paragraph{Centrality}
Classical applications of Closeness Centrality \cite{bavelas1950communication,beauchamp1965} and Betweenness Centrality \cite{freeman1977} are in the area of social network analysis \cite{wasserman1994social}. E.g. Closeness Centrality indicates how fast information spreads and Betweeness Centrality indicates how easy it can be manipulated by certain nodes. However, these findings can often be transferred to networks based on electric devices. One example are simple routing protocols, where malicious nodes with high centrality have the power to manipulate large parts of the communication. Eppstein and Wang showed in \cite{eppstein:2001:fastCCapprox} a fast approximation algorithm to derive the closeness centrality of a node by sampling only a few nodes $S$ and computing shortest paths from $S$ to all nodes. Brandes \cite{brandes:2001:fasterBCalgo} was able to provide a a fast algorithm for computing betweenness centrality recursively, which Bader et al. \cite{bader:2007:BCapprox} extended to an even faster approximation using an adaptive sampling techique.

We continue by reviewing a few applications in network analysis and design independent of social networks: For distance approximation, when using the landmark method \cite{ng2002predicting,pias2003lighthouses}, it is mentioned in \cite{aggarwal2010survey} that a modification of \cite{rattigan2007graph} chooses landmarks using local closeness centrality. According to \cite{shao2013trinity}, this is  implemented e.g. in the Trinity Graph Engine.
A relation to network flows is established e.g. in \cite{borgatti2005centrality} for Closeness as well as Betweenness Centrality. The authors of \cite{daly2007social} propose a routing scheme based on Betweenness Centrality. \todoI{wiki has more stuff on routing}

\paragraph{MRCT}
Distributed approximations of MRCT in the CONGEST model were studied in \cite{hochuli:holzer:MRCST, khan2008efficient, sarma12}. The authors of \cite{khan2008efficient} showed how to use results of \cite{fakcharoenphol2003tight} to obtain a randomized $\BO(\log n)$-approximation of the MRCT in time $\BO\left(D_{sp}\cdot \log^2 \left(n\right)\right)$. Observe that this might be $\BO(n\log^2 n)$ even in a graph with hop diameter $D_h=1$. In \cite{hochuli:holzer:MRCST} two algorithms were presented for unweighted graphs and weighted graphs, where the weight of an edge corresponds to the traversal time of the edge. A deterministic one, that computes a $2$-approximation in linear time $\BO(n)$ and a randomized one, that computes w.h.p. an $2+\varepsilon$-approximation in time $\BO(D+\log n)$ (and variations thereof). A lower bound of $\Omega(\sqrt{n}+D)$ for any randomized $\varepsilon$-error algorithm that computes a $[poly$ $n]$-approximation was shown in \cite{sarma12} for weighted graphs. For a more detailed overview on work related to MRCT computations in various computational models, we refer the reader to \cite{hochuli:holzer:MRCST}.

\paragraph{Shortest Paths}
At the heart of our algorithms is an algorithm to compute $|S|$ directed layered graphs rooted in nodes $S\subseteq V$ in parallel. This algorithm is a modification of the $S$-Shortest Paths algorithm of \cite{holzer2012optimal}. Furthermore our modification extends this $S$-Shortest Paths algorithm to the weighted setting and to compute shortest paths only up to a certain distance that can be provided as input. Note that the $k$-source detection algorithm by Lenzen and Peleg~\cite{lenzen2013efficient} can be extended to weighted graphs in a similar way and could be modified to compute directed layered graphs as well, but we developed our method that resulted in this paper independently and simultaneously~\cite{dissler-thesis,holzer2013phd}.

\section{Model, Definitions and Problems}\label{sec:model}

\paragraph{Model:} Our network is represented by an undirected graph $G=\left(V,E\right)$. Nodes $V$ correspond to processors (computers, sensor nodes or routers),
two nodes are connected by an edge from set $E$ if they can communicate directly with each other. We denote the number of nodes of a graph by $n$, and the number of its edges by $m$. Furthermore we assume that each node has a unique ID in the range of $\{1,\dots,2^{\BO\left(\log n\right)}\}$, i.e. each node can be represented by $\BO(\log n)$ bits. 
Initially the nodes know only $G$ the IDs of their immediate neighborhood, the number of nodes in $G$ and the node with the lowest ID further denoted as node $1$.\footnote{Note that for computing the schedule in our framework we do not need to know $n$ and the node with smallest ID. These assumptions are only necessary for our applications such as centrality computation. The runtime of these applications depends on the diameter of the graph such that computing these values does not affect the runtime.} \\

An unweighted shortest path in $G$ between two nodes $u$ and $v$ is a $u-v$-path consisting of the minimum number of edges contained in any $u-v$-path. Denote by $d\left(u,v\right)$ the unweighted distance between two nodes $u$ and $v$ in $G$, which is the length of an unweighted shortest $u-v$-path in $G$. We also say $u$ and $v$ are $d(u,v)$ hops apart. By $\omega:E\rightarrow \mathbb{N}$ we denote a graph's weight function and by $\omega(e)$ the non-negative weight of an edge in $G$. For an edge $e=(u,v)$ we often write $\omega(u,v)$ to denote $\omega(e)$. For arbitrary nodes $u$ and $v$, by $\omega(u,v):=\min_{\{P|P\text{ is }u-v\text{-path in }G\}}\sum_{e\text{ is edge in }P}\omega(e)$ we define the weighted distance between $u$ and $v$, that is the weight of a shortest weighted path connecting $u$ and $v$. In case we consider a subgraph $H$ of a graph $G$, we denote by $d_H$ and $\omega_H$ the distances with respect to using only edges in $H$.

In our model we consider only symmetric edge weights, i.e. an edge $(u,v)$ has in both directions the same weight $\omega(u,v)=\omega(v,u)$. We denote a graph with weighted edges as $G=(V,E,\omega)$. By $N_k(v)$ we denote the $k$-neighborhood of a node $v$. That is in an unweighted graph all nodes in $G$ that can be reached from $v$ using $k$ hops. And in a weighted graph all nodes that can be reached from $v$ using paths with at most weight $k$.
We use the convention that $v\in N_{k}(v)$. Given a set $S\subset V$, set $N_k(S):=\cup_{v\in S}N_k(v)$ denotes the $k$-neighborhood of $S$. During the paper we denote the degree $|N_1(v)-1|$ of a node $v$ by $deg(v)$.

\begin{definition}
The weighted eccentricity $ecc_{\omega}\left(u\right)$ of a node $u$ is the largest weighted distance to any other node in the graph, that is $ecc_{\omega}\left(u\right):=\max_{v\in V} \omega\left(u,v\right)$.\\
The unweighted (hop) eccentricity $ecc_{h}\left(u\right)$ of a node $u$ is the largest hop distance to any other node in the graph, that is $ecc_{h}\left(u\right):=\max_{v\in V} d\left(u,v\right)$.
\end{definition}

\begin{definition}
The weighted diameter $D_{\omega}\left(G\right):=\max_{u\in V} ecc_{\omega}(u)$ of a graph $G$ is the maximum weighted distance between any two nodes of the graph.
\\ 
The unweighted diameter (or hop diameter) $D_{h}\left(G\right):=\max_{u\in V} ecc_{h}(u)$ of a graph $G$ is the maximum number of hops between any two nodes of the graph.
\end{definition}

Observe that $D_h=D_\omega$ for unweighted graphs. In weighted graphs it is true that $1\leq D_h\leq n$ as well as $D_{h}\leq D_\omega$.

In this paper we need a graph substructure to describe shortest paths from a root node $r$ to all other nodes. Such a representation can be found e.g. in a tree. However, for applications of our framework to e.g. calculating the Betweenness Centrality of a node in Section \ref{sec:BC}, we need to consider all shortest paths between two nodes, not just a single one as provided by a tree.
Therefore we need to define a graph structure which clearly indicates all possible shortest paths between a root node $r\in V$ and all other nodes $v\in V$.

\begin{definition}[Tree $T_r$]
 Given a node $r$, we denote the spanning tree of $G$ that results from performing a breadth-first search ($BFS_r$) starting at $r$ by $T_r$.
\end{definition}

\begin{definition}[Directed Leveled Graph (DLG) $\calL_r$, unweighted]
 Given a node $r\in V$ of a graph $G=(V,E)$, partition all nodes $v\in V$ in disjoint subsets $V_0,V_1,\dots,V_{D_h}$ with $v\in V_i$ if $d(r,v)=i$. Add an directed edge $(u,v)$ to a set $\vec E_r$ for every edge $(u,v)\in E$ with $u\in V_i$ and $v\in V_{i+1}$.
That is, every shortest path from $r$ to any other node $v$ in $G$ is a directed path in $\calL_r=(V,\vec E_r)$. We say $\calL_r$ is rooted in $r$.
A graphic is provided in 
\iffull
Figure \ref{abb:leveledgraph}.
\fi

In the weighted case, we partition all nodes $v\in V$ in disjoint subsets $V_0,V_1,\dots,V_{D_\omega}$ with $v\in V_i$ if $\omega(r,v)=i$. Add an directed edge $(u,v)$ to set $\vec E_r$ for every edge $(u,v)\in E$ with $u\in V_i$ and $v\in V_{i+\omega(u,v)}$.
That is, every shortest weighted path from $r$ to any other node $v$ in $G$ is a directed weighted path in $\calL_r=(V,\vec E_r,\omega)$. 
A graphic is provided in 
\iffull
Figure \ref{abb:leveledgraphweighted}.
\fi

\end{definition}
\iffull
\begin{figure}[htb]
\begin{center}
\includegraphics[width=\linewidth]{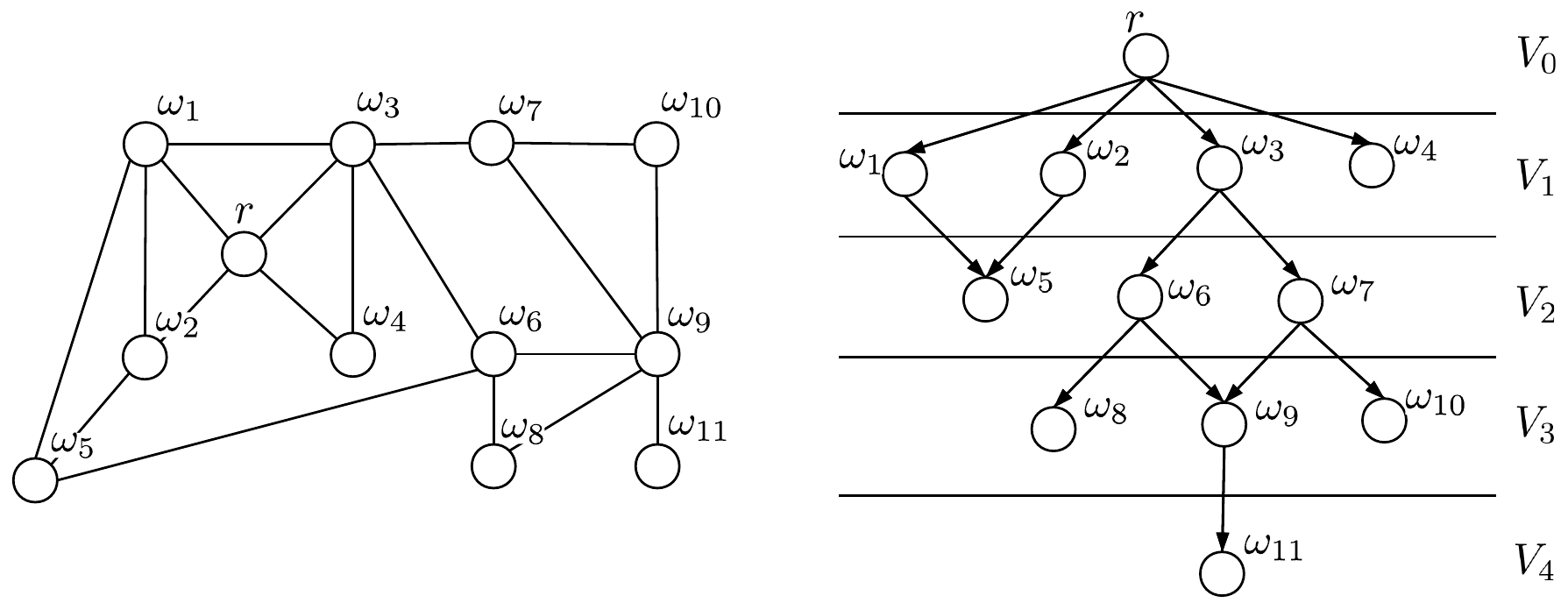}
\caption{$G=(V,E)$ on the left side and its leveled graph $\calL_r$ rooted in $r$ on the right.}\label{abb:leveledgraph}
\end{center}
\end{figure}
\begin{figure}[htb]
\begin{center}
\includegraphics[width=\linewidth]{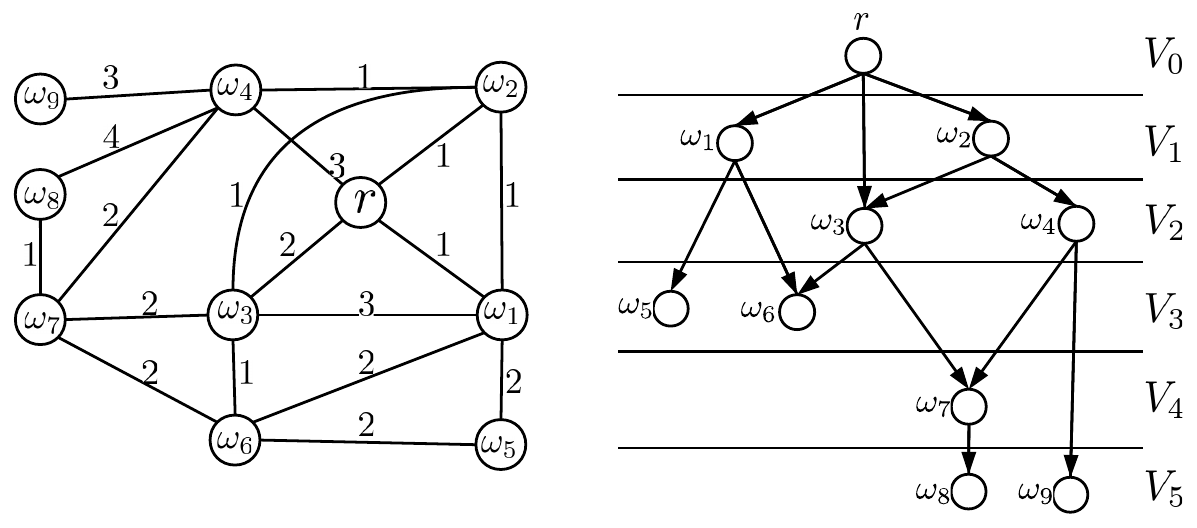}
\caption{$G=(V,E,\omega)$ on the left side and its leveled graph $\calL_r$ rooted in $r$ on the right}\label{abb:leveledgraphweighted}
\end{center}
\end{figure}
\fi

\begin{definition}[parent, children, ancestors and leaves  in a DLG]
 Given a node $v$, we define a parent of $v$ in  a DLG $\calL_r$ as any neighbor $u_p$ of $v$ connected by a directed edge $(u_p,v)\in \vec E_r$ to $v$. As a node can have several parents, we often consider a set of parents, formally defined as: $P_r(v)=\{u\in V:\{u,v\}\in E \text{ and } \omega(r,v)=\omega(r,u)+\omega(u,v)\}$. Children are all neighbors $u_c$ of $v$ connected by a directed edge $(v,u_c)\in \vec E_r$ to $v$. A leaf node in a DLG $\calL_r$ is a node without any children.
\end{definition}

\begin{definition}[approximation] 
Given an optimization problem $P$, denote by $sol_{OPT}$ the value of the optimal solution for $P$ and by $sol_A$ the value of the solution of an algorithm $A$ for $P$.
Let $\epsilon\geq 0$. We say A is a $(+,\epsilon)$-approximation (additive approximation) with additive error $\epsilon$ for $P$ if $|sol_{OPT}-sol_A|\leq\epsilon$ for any input.
Let $\delta\geq 0$. Like in \cite{eppstein:2001:fastCCapprox} we say $A$ is an inverse additive approximation with inverse additive error $\delta$ for $P$ if $\left|\frac{1}{sol_{OPT}}-\frac{1}{sol_A}\right|\leq\delta$ for any input. 
We say $A$ is $\left(\times ,\rho_1\right)$-approximative (a one-sided multiplicative approximation) for $P$ if $(\rho_1\cdot sol_{OPT})\le sol_A \le sol_{OPT}$ and $0<\rho_1<1$ or if $sol_{OPT}\le sol_A \le (\rho_1\cdot sol_{OPT})$ and $\rho_1>1s$ for any input.
\end{definition}
\iffull
Like in \cite{eppstein:2001:fastCCapprox}, the inverse additive error is used in the closeness centrality approximation, see Section \ref{sec:CC}.
\fi
\begin{fact}\label{fact:ecc-approx-diam}
 For any node $u\in V$ we know that $ecc_{\omega}\left(u\right)\leq D_{\omega}\left(G\right) \leq 2\cdot ecc_{\omega}\left(u\right)$.
\end{fact}

\subsection{Problem Statements}
We start by formally stating the problems we consider in this paper. First we state local multi-aggregation, which is not only of interest by itself but turns out to be at the heart of the other problems. Centrality-measures are graph-properties that are based on the number of shortest paths. Besides classical applications in social network analysis \cite{carrington2005models,wasserman1994social}, they have a number of applications in design and analysis of distributed networks such as \cite{aggarwal2010survey,borgatti2005centrality,daly2007social}. Minimum Routing Cost Spanning Trees can be used to minimize the average cost of communication in a network while keeping a sparse routing-structure \cite{hu1974optimum,wong1980worst}. 
 \begin{definition}(local multi-aggregation).
In the problem of local multi-aggregation, we are given a subset $S$ of nodes $V$ in a graph $G$. Each of the nodes in $G$ contains (several) values and each node in $S$ wants to evaluate a (possibly different) aggregation-function based on (some of) these values stored in nodes of its $k$-neighborhood.
We consider two variations:
\begin{enumerate}
\item only one (shortest) $s$-$t$-path is needed for each $(s,t)\in S\times N_k$ to perform the aggregation.
\item all shortest $s$-$t$-paths are needed for each $(s,t)\in S\times N_k$ to perform the aggregation.
\end{enumerate}\end{definition}

Betweenness centrality is a measure for centrality of a node $u\in V$, which is based on the number of shortest paths in a graph $G$ node $u$ is part of. In a graph $G=(V,E)$, let $\sigma_{st}$ denote the number of shortest paths from $s$ to $t$ and let $\sigma_{st}(v)$ denote the number of shortest paths from $s$ to $t$ that go through the node $v$ for $s,t,v\in V$. Closeness centrality is another measure to identify important nodes in a graph. The closeness centrality of a node is the inverse of the average distance to all other nodes. At the heart of our solutions for these problems is the $S$-SP problem and we define:

\begin{definition}[Betweenness Centrality, as stated in \cite{freeman1977}] $BC(v):=\sum\limits_{s\neq v\neq t \in V}\frac{\sigma_{st}(v)}{\sigma_{st}}$.
\end{definition}

\begin{definition}[Closeness Centrality, as stated in \cite{beauchamp1965}] $CC(v):=\frac{n-1}{\sum_{i=1}^nd(v,i)}$.
\end{definition}

\begin{definition}\label{def:MRCT}(Minimum Routing Cost Tree (MRCT) as stated in \cite{WuLBCRT99} ).
The MRCT problem~\cite{WuLBCRT99} is defined as follows. The routing cost $RC\left(H\right)$ of a subgraph $H$ of $G$ is the sum of the routing costs of all pairs of vertices in the tree, i.e., $RC\left(H\right):=\sum_{u, v\in V} \omega_H\left(u, v\right)$. Our goal is to find a spanning tree with minimum routing cost.
Given a subset $S$ of the vertices $V$ in $G$, in the $S$-MRCT problem  \cite{hochuli:holzer:MRCST},\ our goal is to find a subtree $T$ of $G$ that spans $S$ and has minimum routing cost with respect to $S$. That is a tree $T$ such that $RC_S\left(T\right):=\sum_{u, v\in S} \omega_T\left(u, v\right)$ is minimized. Here, $RC_S\left(T\right)$ denotes the routing cost of $T$ with respect to $S$.
\end{definition}

\begin{definition}[$S$-SP]
 Let $G = (V, E)$ be a graph. In the $S$-Shortest Paths ($S$-SP) problem, we are
given a set $S\subseteq V$ and need to compute the shortest path lengths between any pair of nodes in $S\times V$ such that in the end each node in $ V$ knows its distance to all nodes in $S$.
\end{definition}

\section{Local Multi-Aggregation}
\label{sec:aggAlgo}
Many distributed algorithms to compute network properties can be stated in a way that they distribute and/or aggregate information through a tree or a directed leveled graph (DLG). For example computing the max-value of the $k$-neighborhood (using a tree, see the text below), the sum of all values in the $k$-neighborhood (using a tree, see the text below) and computing centrality measures (using a DLG, see Sections \ref{sec:BC} and \ref{sec:CC}).

By extending the algorithms of \cite{holzer2012optimal} and \cite{hochuli:holzer:MRCST}, we present an algorithm which can aggregate information in the $k$-Neighborhood of a set of root nodes $r\in S$. Note that if information from the whole graph is needed, $k$ can be set to $D_\omega$.

All known exact computations of these properties need to compute a tree or DLG for every node $v\in V$, to evaluate $v$'s dependency on other nodes.
Often the exact computation can be approximated by evaluating only the dependencies of a subset $S\subseteq V$ of nodes and therefore computing only $|S|$ trees or DLGs with depth $k$ in respect to the root nodes $r\in S$.
We provide in Section \ref{sec:aggregate} an algorithm to compute $|S|$ such DLGs rooted in the nodes of $S$ in time $\BO(|S|+k)$.
Furthermore we provide an algorithm to aggregate information on these previously computed DLGs, again with a time complexity of $\BO(|S|+k)$.
Both algorithms are able to execute additional algorithms, ${\mathcal{A}}_{agr}$ and ${\mathcal{A}}_{agr}'$ respectively within each node, to distribute and aggregate information about the desired network property, where the network property description replaces the index $agr$. These functions ${\mathcal{A}}_{agr}$ and ${\mathcal{A}}_{agr}'$ can be chosen depending on the aggregating problem at hand.
A possible application with choices of ${\mathcal{A}}_{max}$ and ${\mathcal{A}}_{max}'$ to compute the max-value in the $k$-neighborhood of each root node $r\in S$ is shown in 
\iffull
Example \ref{ex:maxval}.

\begin{example}
(Computing the max-value of each k-Neighborhood of a set of nodes $r\in S$, $S\subseteq V$, with our proposed algorithm.)
Each node $r\in S$ starts Algorithm\iffull~\ref{alg:DLGcomp}\fi~\textsc{DLGcomp}. 
While spanning the DLGs $\calL_r$ rooted in each node $r\in S$ in Algorithm\iffull~\ref{alg:DLGcomp}\fi~\textsc{DLGcomp} no algorithm $\mathcal{A}_{max}$ needs to be specified. In Algorithm\iffull~\ref{alg:DLGagr}\fi~\textsc{DLGagr} the maximum node value needs to be aggregated. Therefore we define algorithm $\mathcal{A}_{max}'$ as follows: $\mathcal{A}_{max}'$ initializes on each node $u\in V$ for each root node $r$ a variable $max_v[r]$ to store the highest value received from any child in $\calL_r$, first the node's own value $val(u)$ is stored in $max_v[r]$ for every $r\in S$. On execution of $\mathcal{A}_{max}'$ with message $msg_c$ and ID $r$ from a child, $\mathcal{A}_{max}'$ stores the value of $msg_c$ in $max_v[r]$ if $val(msg_c)>max_v[r]$. Then $max_v[r]$ is stored in $msg_p$ and sent to the parents of $u$.
After the execution of Algorithm\iffull~\ref{alg:DLGcomp}\fi~\textsc{DLGcomp} and\iffull~\ref{alg:DLGagr}\fi~\textsc{DLGagr} each root node $r$ stores in $max_v[r]$ the max-value of its own k-Neighborhood.
\label{ex:maxval}
\end{example}

\fi
 Or  ${\mathcal{A}}_{sum}$ and ${\mathcal{A}}_{sum}'$ to compute the sum of all values in the $k$-neighborhood of each root node $r\in S$ is provided in 
\iffull
Example \ref{ex:sumval}.

\begin{example}\label{ex:sumval}
(Computing the sum of all values in each k-Neighborhood of a set of nodes $r\in S$, $S\subseteq V$, with our proposed algorithm.)
Each node $r\in S$ starts Algorithm\iffull~\ref{alg:DLGcomp}\fi~\textsc{DLGcomp}. 
While spanning the DLGs $\calL_r$ rooted in each node $r\in S$ in Algorithm \iffull~\ref{alg:DLGcomp}\fi~\textsc{DLGcomp} no algorithm $\mathcal{A}_{max}$ needs to be specified. In Algorithm\iffull~\ref{alg:DLGagr}\fi~\textsc{DLGagr} the sum of all node values in the k-Neighborhood needs to be aggregated. Therefore we define algorithm $\mathcal{A}_{sum}'$ as follows: $\mathcal{A}_{sum}'$ initializes on each node $u\in V$ for each root node $r$ a variable $sum_v[r]$ to store the sum of all values received from any child in $\calL_r$ plus its own value $val(u)$. First the node's own value $val(u)$ is stored in $sum_v[r]$ for every $r\in S$. On execution of $\mathcal{A}_{sum}'$ with message $msg_c$ and ID $r$ from a child, $\mathcal{A}_{sum}'$ adds the value of $msg_c$ to $sum_v[r]$. Then $sum_v[r]$ is stored in $msg_p$ and sent to (only) one parent of $u$. If node $u$ has more than one parent, the message is sent to the parent node with the lowest ID. This results in the Tree Variation in Section \ref{sec:tree}.
After the execution of Algorithm\iffull~\ref{alg:DLGcomp}\fi~\textsc{DLGcomp} and\iffull~\ref{alg:DLGagr}\fi~\textsc{DLGagr} each root node $r$ stores in $sum_v[r]$ the sum of all values in its own k-Neighborhood.
\end{example}

\fi
\begin{remark}[Time optimality of Algorithms\iffull~\ref{alg:DLGcomp}\fi~\textsc{DLGcomp} and\iffull~\ref{alg:DLGagr}\fi~\textsc{DLGagr}]\label{rem:lowerbound}
In a setting where edge-weights correspond to the transmission time through the corresponding edge, any solution for an algorithm to aggregate information from all nodes $u\in V$ to one root node $r$ needs at least $\Omega(D_\omega)$ time. Now consider the graph of Figure~\ref{fig:lowerbound}. There we have $|S|$ root nodes $1,\dots,|S|$ connected to a chain of $D-1$ nodes and assume this chain determines the diameter $D_\omega$. Assume each of the chain nodes stores $|S|$ different values that are encoded using $\BO(\log n)$ bits each. Assume each root node $i$ computes an aggregation function based on the $i$'th value of each chain node. Due to congestion, the time until all $|S|$ values of the chain node at distance $D_\omega$ to the root nodes arrives at the corresponding root nodes is $|S|+D_\omega$. This can be extended to $|S|+k$ in the unweighted case with $k<D$.
\end{remark}
\begin{figure}[ht]
	\begin{center}
		\includegraphics[width=.5\textwidth]{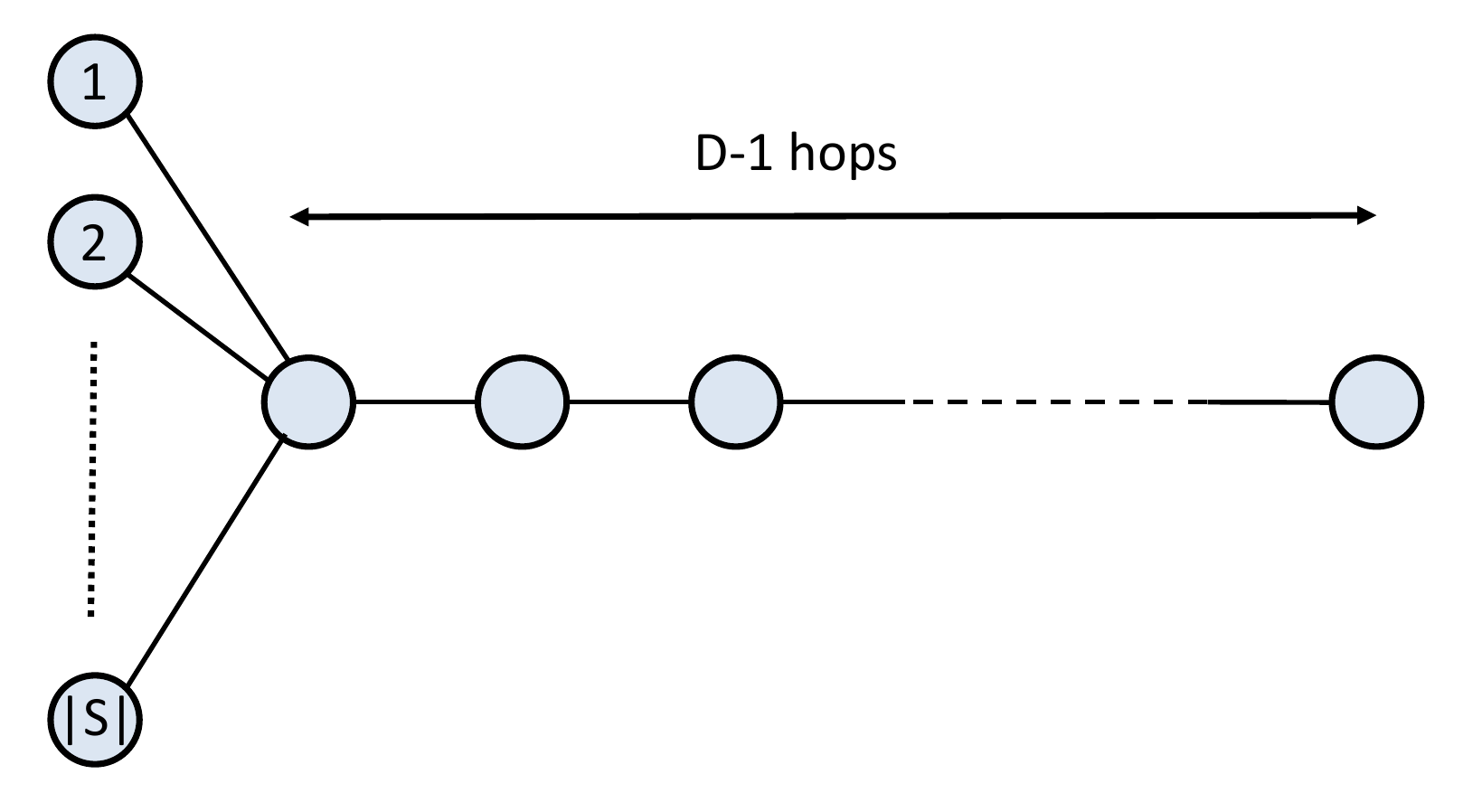}
	\end{center}
	\caption{Graph based on the construction used in Remark~\ref{rem:lowerbound}.}
	\label{fig:lowerbound}
\end{figure}

\subsection{Algorithm for Local Multi-Aggregation}\label{sec:aggregate}

First we describe Algorithm\iffull~\ref{alg:DLGcomp}\fi~\textsc{DLGcomp} which computes $|S|$ times a DLG with depth $k$, one for each node in $S$, and can distribute/broadcast information (depending on the aggregation function at hand) along these computed DLGs. In Algorithm\iffull~\ref{alg:DLGcomp}\fi~\textsc{DLGcomp} we extend the $S$-SP Algorithm of \cite{holzer2012optimal}, which computes shortest paths. (Note that alternatively one could also modify the $k$-source detection algorithm by Lenzen and Peleg~\cite{lenzen2013efficient}.) In Algorithm\iffull~\ref{alg:DLGagr}\fi~\textsc{DLGagr}, we aggregate information efficiently through the computed DLGs.
Algorithm\iffull~\ref{alg:DLGcomp}\fi~\textsc{DLGcomp} is designed as a subroutine and has two input parameters $|S|$ and $k$ as well as a function ${\mathcal{A}}_{agr}$.
Parameter $|S|$ is the number of DLGs $\calL_r$ rooted in $r\in S$ which Algorithm\iffull~\ref{alg:DLGcomp}\fi~\textsc{DLGcomp} has to compute.
Parameter $k$ denotes the depth of the $k$-neighborhood.

And function ${\mathcal{A}}_{agr}$ specifies an algorithm executed along with the computation of the $|S|$ DLGs.
Further a node $r\in S$ knows that it is in $S$ (e.g. as it is sampled by some procedure or interested in performing an aggregation by itself).
\todoI{describe basic short idea of Algo}
Lets consider first only the special case where we compute a single DLG $\calL_r$ rooted in node $r$. We denote by \emph{$r$-message} any messages sent during an execution of Algorithm\iffull~\ref{alg:DLGcomp}\fi~\textsc{DLGcomp} that belongs to the construction of the DLG rooted in node $r$ (later we keep this term for different roots $r\in S$).  Node $r$ starts with sending an $r$-message to all its neighbors $w$ and in the next time slot, those neighbors $w$ send an $r$-message to their neighbors, except to $r$.
A DLG $\calL_r$ continues to grow as follows: in time slot $t$ all nodes $u$ at distance $\omega(r,u)=t$ receive a $r$-message from their neighbors $w$ with distance $\omega(r,w)=t-1$. If the first $r$-message is received from neighbor $w_1$ of node $u$, $w_1$ is considered a parent of $u$ in $\calL_r$. However, if in the same time slot $u$ receives further $r$-messages from different neighbors $w_i$, those are considered parents too. One time slot after receiving the first $r$-message(s), node $u$ computes an $r$-message and sends it to all neighbors $w_i$ which are not considered a parent of $u$.

Now we say a $r$-message received by $u$ is considered \emph{valid}, if it is sent from a parent of $u$ in $\calL_r$. Each $r$-message contains the ID of the root node $r$ and the weighted distance $\omega_r:=d_\omega(r,u)$ from $r$ to the receiving node $u$.

To simulate a weighted graph, the messages get delayed 
\iffull in Line \ref{line:li} \fi
 corresponding to the corresponding edge weight $\omega(u,v)$, if a message is to be sent from node $u$ to $v$. For unweighted graphs, the variable $l_i$
\iffull (Line \ref{line:li}) \fi
 contains the lowest ID in $L_i$ which is not in $L_{delay}$.\todoI{details: what is $L_i$, $L_{delay}$.}

While computing $\calL_r$, the ID $r$ and some additional information is stored and sent along with the $r$-messages. Each node stores time $t$ and parent $w$ from which it received a $r$-message
\iffull
 (Line \ref{line:tau} and \ref{line:tauplusone})
\fi
. This is needed to efficiently execute Algorithm\iffull~\ref{alg:DLGagr}\fi~\textsc{DLGagr} later on. 
A node $u$ stores in schedule $\tau[r][h]$ and in $parent\_in\_T_{r_h}$ arrival times and parent nodes $u_h$ if $u\in N_k(r)$ \iffull(Lines \ref{line:kvarstart} to \ref{line:exf} are only executed if $u\in N_k(r_i)$, i.e. if $\omega_i\leq k$, where $\omega_i$ was received in a $r_i$-message)\fi.
Nonetheless, $r_i$-messages are still forwarded, even if the distance $\omega_i$ in a $r_i$-messages is larger than $k$.
This is needed to ensure that $r$-messages are properly delayed.

Furthermore each node $u$ executes algorithm ${\mathcal{A}}_{agr}$\iffull (Line \ref{line:exf})\fi, and can be used to distribute information in $\calL_r$. Algorithm ${\mathcal{A}}_{agr}$ is executed with parameters $(r,msg_p,u_h)$ each time $u$ receives a valid $r$-message from neighbor $u_h\in \{u_1,\dots,u_{deg(u)}\}$. Message $msg_p$ is the part of the $r$-message computed by algorithm ${\mathcal{A}}_{agr}$ on a parent node $u_h$ of $u$ in $\calL_r$.
After the execution for every parent  $u_h$ of $u$ in $\calL_r$ with the corresponding $msg_p$,  algorithm ${\mathcal{A}}_{agr}$ computed and stored in $msg_c[r]$ the message which is sent to the children of $u$ in $\calL_r$\iffull (Line \ref{line:sendreceivedist})\fi. Furthermore algorithm ${\mathcal{A}}_{agr}$ is executed once on every node for definition and initialization of additional global variables on the node\iffull (Line \ref{line:DLGinitf})\fi.

With Algorithm\iffull~\ref{alg:DLGcomp}\fi~\textsc{DLGcomp}, multiple leveled graphs $\calL_{r_1},\dots, \calL_{r_{|S|}}$ start growing in the same time slot. One $\calL_r$ for each $r\in S$. This could lead to congestions. To prevent that, every time two Algorithm\iffull~\ref{alg:DLGcomp}\fi~\textsc{DLGcomp} messages cross\footnote{That means both messages are received in the same time slot by the same node or both are sent in the same time slot through the same edge $(u_h,u)$, one from $u$ to $u_h$ and one in the opposite direction.}, the Algorithm\iffull~\ref{alg:DLGcomp}\fi~\textsc{DLGcomp} message with the higher ID is delayed one time slot. 
This is done\iffull in Line \ref{line:DLGputLdelay}\fi by putting the higher ID into queue $L_{delay}$ or by retransmitting the message in the next time slot through the same edge $(u_h,u)$ again, respectively\iffull (due to the if-statement in Line \ref{line:DLGdelay})\fi.

Similar as in the proof of the S-SP algorithm \cite{holzer2012optimal}, we show that the total delay of any Algorithm\iffull~\ref{alg:DLGcomp}\fi~\textsc{DLGcomp} is at most $\BO(|S|)$ when $|S|$ many DLGs are constructed in parallel. And we show that despite of the delays we still construct valid leveled graphs.

\iffull
\begin{algorithm}[tp]
\small
\begin{algorithmic}[1]
\STATE $L:=\emptyset$,
\STATE $L_{delay}:=\emptyset$;
\STATE $\tau:=\{\infty,\infty,\dots,\infty\}\times[n]$; \COMMENT{$\tau[v][i]=$ time when message of DLG $\calL_v$ reaches $u$, for each parent $w_1,\dots,w_{\max}\in P_v[u]$ enumerated in $i $}\label{line:tauinit}
\STATE $sp:=\{0,0,\dots,0\}$; \COMMENT{sp[v] := number of parents in shortest paths from $v$ to $u$}
\STATE $\omega:=\{\infty,\infty,\dots,\infty\}$; \COMMENT{$\omega[v]=\omega(v,u)$, $\omega[\emptyset]:=\infty$}\label{line:omegainit}
\IF{$u \in S$}
        \STATE $L:=\{u\}$;
        \STATE $\omega[u]:=0$;
        \STATE $\tau[u][1]:=0$;
\ENDIF
\STATE $L_1,\dots,L_{deg(u)}:=L$;
\STATE \textbf{initialize} algorithm ${\mathcal{A}}_{agr}$;\label{line:DLGinitf}
\FOR{$t=0,\dots,|S|+k$}\label{line:DLGdisploop}
        \FOR{$i=1,\dots,deg(u)$}
                \STATE $l_i:=\min\left\{v\in L_i : \tau[v][sp[v]]+\omega(u,v)\ge t \text{ and }\ v\notin L_{delay}\right\}$; // smallest ID that is not delayed and ready to be scheduled \label{line:li}
        \ENDFOR
        \STATE within one time slot:\label{line:sendreceivedist}
        \newline \textbf{send} $\left(l_1,\omega[l_1]+\omega\left(u,u_1\right), msg_c[l_1]\right)$ to neighbor $u_1$, receive $\left(r_1,\omega_1, msg_{p,1}\right)$ from $u_1$;
        \newline \textbf{send} $\left(l_2,\omega[l_2]+\omega\left(u,u_2\right), msg_c[l_2]\right)$ to neighbor $u_2$, receive $\left(r_2,\omega_2, msg_{p,2}\right)$ from $u_2$;
\todoI{$\omega_2$ vs. $w_i$}        \newline $\dots$
        \newline \textbf{send} $\left(l_{deg(u)},\omega[l_{deg(u)}]+\omega\left(u,l_{deg(u)}\right), msg_c[l_{deg(u)}]\right)$ to neighbor $u_{deg(u)}$, receive $\left(r_{deg(u)},\omega_{deg(u)}, msg_{p,deg(u)}\right)$ from $u_{deg(u)}$;
        
        \STATE $R:=\{r_i:r_i<l_i\text{ and }i\in \{1,\dots,deg(u)\}\}\setminus L$;
        \STATE $s:=min(L_{delay})$;
        \IF{$s\leq \min(R) \textbf{ and } s<\infty$}
                \STATE $L_{delay}:=L_{delay}\setminus\{s\}$;
        \ENDIF
        
        \FOR{$i=1,\dots,deg(u)$}
                \IF[$T_{l_i}$'s message will be delayed due to $T_{r_i}$.]{$r_i<l_i$}\label{line:DLGdelay}
                        
                        \IF{$\omega_i\leq \omega[r_i]$} \label{line:treevarstart}
                                                \IF{$\omega_i\leq k$} \label{line:kvarstart}
                                                        \STATE $sp[r_i]=sp[r_i]+1$;
                                                        \STATE $\tau[r_i][sp[r_i]]:=t$;\label{line:tau}
                                                        \STATE $parent\_in\_T_{r_i}[sp[r_i]]:=$ neighbor $u_i$; \label{line:tauplusone}
                                                        \STATE \textbf{execute} algorithm ${\mathcal{A}}_{agr}(r_i, msg_{p,i})$; \label{line:exf}
                                                \ENDIF
                                                \IF{$r_i\notin L$}
                                                        \STATE $\omega[r_i]=\omega_i$; \label{line:omegacomp}
                                                        \STATE $L:=L\cup\{r_i\}, L_1:=L_1\cup\{r_i\}, L_2:=L_2\cup\{r_i\},\dots,L_{i-1}:=L_{i-1}\cup\{r_i\},L_{i+1}:=L_{i+1}\cup\{r_i\},\dots,L_{deg(u)}:=L_{deg(u)}\cup\{r_i\}$;
                                                        \IF{$\min(R)<r_i \textbf{ or } s<r_i$}
                                                                \STATE $L_{delay}:=L_{delay}\cup\{r_i\}$; \label{line:DLGputLdelay}
                                                        \ENDIF
                                                \ELSE
                                                        \STATE $L_i:=L_i\setminus \{r_i\}$;
                                                \ENDIF
                                        \ENDIF \label{line:treevarend}        
                \ELSE
                        \STATE $L_i:=L_i\setminus \{l_i\}$; \COMMENT{$T_{l_i}$'s message was successfully sent to neighbor $u_{i}$.}
                \ENDIF
        \ENDFOR
\ENDFOR
\end{algorithmic}
\caption{\textsc{DLGcomp} compues $|S|$ DLGs in the k-Neighborhood of the root nodes $r\in S$ \newline
(executed by node $u$) \newline 
        \textbf{Input:} $|S|$, $k$, ${\mathcal{A}}_{agr}$\newline 
        \textbf{parameters passed to ${\mathcal{A}}_{agr}(v, msg_p, u_h)$:} message $msg_p$ of a parent $u_h$ of $u$ in DLG $\calL_v$}
\label{alg:DLGcomp}
\end{algorithm}
\fi

In  Algorithm\iffull~\ref{alg:DLGagr}\fi~\textsc{DLGagr}, information gets aggregated according to Algorithm $\mathcal{A}_{agr}'$. To do this, the computed leveled graphs of each root node $r_j$ get processed in a bottom-up fashion. 
Algorithm\iffull~\ref{alg:DLGagr}\fi~\textsc{DLGagr} has three inputs: the number of DLGs $|S|$, the depth $k$ of the $k$-neighborhood, which are needed to bound the runtime, and an algorithm $\mathcal{A}_{agr}'$, which can consist of an initialization part and a computation part. The initialization part is executed once on every node $u\in V$ starting the computation\iffull(loop of Algorithm\iffull~\ref{alg:DLGagr}\fi~\textsc{DLGagr}, Line \ref{line:DLGloopagg})\fi. In the loop in time slot $t=|S|+k-\tau[r_j][h]$ a node $u$ sends the information computed by $\mathcal{A}_{agr}'$ to its parent $u_h\in \{u_1,\dots,u_{deg(u)}\}$ in $\calL_{r_j}$. Here schedule $\tau[r_j][h]$ is the time when $u$ received a valid $r_j$-message from neighbor $u_h$. The schedule $\tau$ was computed before in Algorithm\iffull~\ref{alg:DLGcomp}\fi~\textsc{DLGcomp}\iffull (Line \ref{line:tau})\fi. 

The computation part of $\mathcal{A}_{agr}'(r,msg_{c,u_h},u_h)$ is executed on a node $u\in V$ each time $u$ received a message $msg_{c,u_h}$ from a child node $u_h$ in $\calL_r$. When $u$ has received all messages $msg_{c,u_h}$ form all children $u_h$ in $\calL_r$ and $u$ has executed $\mathcal{A}_{agr}'$ once for each of those messages, $\mathcal{A}_{agr}'$ computed and stored in $msc_p[r]$ the information that is subsequently sent to all parents of $u$ in $\calL_r$.

\iffull
\begin{algorithm}[ht]
\begin{algorithmic}[1] 
\STATE \textbf{initialize} algorithm ${\mathcal{A}}_{agr}'$; // and use same variables as Algorithm\iffull~\ref{alg:DLGcomp}\fi~\textsc{DLGcomp}
\FOR{$t=0,\dots,|S|+k$}\label{line:DLGloopagg}
        \STATE within one time slot:\label{line:sendreceiveagg}
        \newline
                          \hspace*{0.5cm}\textbf{foreach} $v\in L$, $i\in \{1,\dots,sp[v]\}$ \textbf{such that} $t=|S|+k-\tau[v][i]$ \textbf{do}\\
\hspace*{1cm}\textbf{send} $\left(v,msg_p[v]\right)$ to $parent\_in\_T_{v}[i]$;
        \newline \hspace*{0.5cm}\textbf{receive} $\left(v_1,msg_{c,u_1}\right)$ from neighbor $u_1$;
        \newline \hspace*{0.5cm}\textbf{receive} $\left(v_2,msg_{c,u_2}\right)$ from neighbor $u_2$;
        \newline \hspace*{0.5cm}$\cdots$
        \newline \hspace*{0.5cm}\textbf{receive} $\left(v_{deg(u)},msg_{c,u_{deg(u)}}\right)$ from neighbor $u_{deg(u)}$;
        \FOR{$i=1,\dots,deg(u)$}
                \IF{$v_i\neq$ empty}
                        \STATE \textbf{execute} algorithm ${\mathcal{A}}_{agr}'(v_i, msg_{c,u_i})$;\label{line:exg}
                \ENDIF
        \ENDFOR
\ENDFOR
\end{algorithmic}
\caption{\textsc{DLGagr} computes aggregation function (executed by node $u$) \newline
        \textbf{Input:} $|S|$, $k$, ${\mathcal{A}}_{agr}'$  \newline
        \textbf{parameters passed to ${\mathcal{A}}_{agr}'(v,msg_c, u_h)$:} message $msg_c$ of a child $u_h$ of $u$ in DLG $\calL_v$}
\label{alg:DLGagr}
\end{algorithm}
\fi

\begin{lemma}
Algorithm\iffull~\ref{alg:DLGcomp}\fi~\textsc{DLGcomp} computes $|S|$ valid directed leveled graphs, in time $\BO(|S|+D_\omega)$ for $k=D_\omega$.
\label{lemma:algdisp}
\end{lemma}

\iffull
\begin{proof}
The proof is very similar to the proofs for the S-SP algorithm in \cite{holzer2012optimal}. For the \textit{correctness} we state a slightly adapted version of Lemma 12 in \cite{holzer2012optimal}.

\textbf{Correctness: }
First, assume that $S=\{r\}$ and consider the computation of $\calL_r$. At time t, each node $u$ at weighted distance $t$ from $r$ receives a $r$-message from all of $u$'s neighbors $u_h$ that are at distance $t-\omega(u_h,u)$ to $r$. All edges incident to neighbors $u_h$ that sent such a message are added to DLG $\calL_r$, directed from $u_h$ to $u$.

Now consider the case, where the set $S\supset\{r\}$ contains at least two nodes.
We analyze how the computations of other nodes affect the computation of $\calL_r$. 
We say that a $r$-message has \textit{`reached'} a node $u$ through edge $(u_h, u)$ from a neighbor $u_h$ in time slot t, if the message was successfully received and has been removed from the delay queue $L_{delay}$ in time slot $t$, or it was successfully received in time slot $t$ and not put into the queue $L_{delay}$ at all.
It turns out that the first $r$-messages of $\calL_r$'s computation which \textit{reach} $u$ are transmitted through the edges at distance $t-\omega(u_h,u)$ as in the case before. Consider two neighbors $u_1,u_2\in N_1(u)$ of $u$; we can ignore the case that $u$ has only one neighbor since this case trivially satisfies our claim.
A $r$-message containing weighted distance $\omega(r,u_j)$, has \textit{reached} $u$ through the edge $(u_1,u)$ earlier than through edge $(u_2,u)$ if and only if $\omega(r,u_1)<\omega(r,u_2)$.
We show this by proving that the set of messages with lower IDs which delay the $r$-messages is the same for both paths $(r,u_1,u)$ and $(r,u_2,u)$: Assume that the computation of $\calL_i$ is delaying the $r$-message sent through $(r,u_1,u)$ at some point.
Then the $i$-message reaches, in case the $i$-message is coming from $u$'s direction, node $u$ earlier than the $r$-message.
Even if an $i$-message and a $r$-message are transmitted in the same time slot to $u$ (through $(u_1,u)$ and $(u_2,u)$, respectively), the $i$-message delays the $r$-message in the node by putting it into the queue $L_{delay}$, and the $i$-message \textit{reaches} $u$ earlier.
Thus the $i$-message also delays the $r$-message running through path $(r,u_2,u)$, if it did not already delayed it earlier.

\textbf{Runtime: }
In case $S=\{r\}$ contains only one element, the computation of  DLG $\calL_r$, which is not delayed by other computations, takes at most $D_\omega'$ time steps. Because in time step $t$ all nodes $u\in V$ at distance $\omega(r,u)=t$ receive a $r$-message, as a consequence in time step $D_\omega$ every node $v\in V$ with distance $\omega(r,v)\leq D_\omega$ receives or has received a $r$-message. Since there are no nodes $u\in V$ with $\omega(r,v)> D_\omega$, the computation of $\calL_r$ stops after $D_\omega'$ time steps.

We showed that a $r$-message can be delayed at most one time slot by another $i$-message with $i<r$, and therefore by at most $|S|$ time slots by all other $|S|$ DLG-constructions. Thus Algorithm\iffull~\ref{alg:DLGcomp}\fi~\textsc{DLGcomp} runs in $\BO(|S|+D_\omega)$.
\end{proof}
\fi

\begin{lemma}
Algorithm\iffull~\ref{alg:DLGagr}\fi~\textsc{DLGagr} aggregates information through $|S|$ directed leveled graphs, in time $\BO(|S|+D_\omega)$ for $k=D_\omega$.\label{lemma:algagg}
\end{lemma}

\iffull
\begin{proof}
\textbf{Runtime: }
In Algorithm\iffull~\ref{alg:DLGagr}\fi~\textsc{DLGagr}, a node $u$ sends only messages to parent nodes in all DLGs $\calL_r$ $\forall r\in S$. The schedule to send the messages is the same as in Algorithm\iffull~\ref{alg:DLGcomp}\fi~\textsc{DLGcomp} stored in $\tau$, but reversed. Therefore the runtime of Algorithm\iffull~\ref{alg:DLGagr}\fi~\textsc{DLGagr} is bound by the runtime of Algorithm\iffull~\ref{alg:DLGcomp}\fi~\textsc{DLGcomp}, which is $\BO(|S|+D_\omega)$.
The complexity of executing algorithms $\mathcal{A}_{agr}$ and $\mathcal{A}_{agr}'$ inside of a node causes no additional communication. Since we are only interested in communication complexity, the total runtime is $\BO(|S|+D_\omega)$.

\textbf{Correctness: }
Since we use the schedule $\tau$ of Algorithm\iffull~\ref{alg:DLGcomp}\fi~\textsc{DLGcomp} in reverse order to send messages, it is guaranteed that a node $u$ receives all messages from all children in a DLG $\calL_r$ before $u$ sends the first message to any parent in $\calL_r$.
\end{proof}
\fi

\begin{lemma}
Algorithms\iffull~\ref{alg:DLGcomp}\fi~\textsc{DLGcomp} and\iffull~\ref{alg:DLGagr}\fi~\textsc{DLGagr} aggregate information in $\BO(|S|+k)$ time for $k\leq D_\omega$.
\label{lem:kneighborhood}
\end{lemma}

\iffull
\begin{proof}
\textbf{Runtime: }
In the k-Neighborhood restricted Algorithm\iffull~\ref{alg:DLGcomp}\fi~\textsc{DLGcomp} a $r$-message which is not delayed by other IDs, needs $k$ time slots to reach all nodes $u\in N_k(r)$.
We showed in Lemma \ref{lemma:algdisp} that a $r$-message can be delayed at most $|S|$ time slots due to other ID's.
Hence, to reach all nodes in its $k$-neighborhood an $r$-message needs at most $|S|+k$ time slots.
For aggregating information Algorithm\iffull~\ref{alg:DLGagr}\fi~\textsc{DLGagr} uses schedule $\tau$, which was computed in the k-Neighborhood restricted Algorithm\iffull~\ref{alg:DLGcomp}\fi~\textsc{DLGcomp} and therefore is bound by the same number of time slots $|S|+k$.

\textbf{Correctness: }
As long as a distance $\omega_i$ in a $r$-message is less or equal to $k$, the message is processed the same way as in the restricted Algorithm\iffull~\ref{alg:DLGcomp}\fi~\textsc{DLGcomp} ($k=D_\omega$).
When $\omega_i$ is larger than $k$, then $r$-messages with ID $r$ do no longer extend the DLG $\calL_r$, but just delay all $i$-message with higher ID ($i>r$) the same way as in the restricted Algorithm\iffull~\ref{alg:DLGcomp}\fi~\textsc{DLGcomp}.
As a consequence after the execution of the k-Neighborhood restricted Algorithm\iffull~\ref{alg:DLGcomp}\fi~\textsc{DLGcomp}, all nodes $u\in N_k(r)$ have the same information of DLG $\calL_r$ stored as if the restricted Algorithm\iffull~\ref{alg:DLGcomp}\fi~\textsc{DLGcomp} has been executed.
All nodes $v\notin N_k(r)$ that still send $r$-messages do not contribute to $\calL_r$'s aggregation but are necessary to ensure that the other tree's computations are delayed and trees are constructed correctly. 

By executing Algorithm\iffull~\ref{alg:DLGagr}\fi~\textsc{DLGagr} a node $u$ sends aggregation messages $msg_p[v]$ back to a root node $r$ if and only if $u\in N_k(r)$, as mentioned above. Therefore Algorithm\iffull~\ref{alg:DLGagr}\fi~\textsc{DLGagr} just aggregates information of the neighborhood $N_k(r)$ to a root node $r$.
\end{proof}
\fi

\begin{proof}[Proof of Theorem \ref{theo:DLGruntime}]
Follows by combining Lemma \ref{lemma:algdisp} and \ref{lemma:algagg} and Lemma \ref{lem:kneighborhood}.
\end{proof}

\subsection{Tree Variation}
\label{sec:tree}

In Algorithm\iffull~\ref{alg:DLGagr}\fi~\textsc{DLGagr} we can aggregate information along all shortest $r$-$u$-paths between a root node $r\in S$ and a node $u\in V$.
For some applications it is desirable to aggregate information only along one shortest $r$-$u$-path, as for example in max-value/average-value aggregation or in a  $S$-Shortest Paths to approximate problems such as closeness centrality.
That means a node $u$ sends a message only to one parent in a DLG $\calL_r$ while executing Algorithm\iffull~\ref{alg:DLGagr}\fi~\textsc{DLGagr}.
The result is the same as when aggregating information through a tree $T_r$ rooted in $r$ (instead of a DLG $\calL_r$).
\todoI{Can be done inside aggregation fkt?}
For completeness we show an adaptation of Algorithms\iffull~\ref{alg:DLGcomp}\fi~\textsc{DLGcomp} and\iffull~\ref{alg:DLGagr}\fi~\textsc{DLGagr} which computes and aggregates along trees instead of DLGs.
We provide a detailed description of the Tree Variation in Appendix 
\iffull
\ref{app:tree}.
\fi

\section{Applications}

\subsection{Betweenness Centrality}
\label{sec:BC}

Brandes showed in \cite{brandes:2001:fasterBCalgo}, an algorithm for computing betweenness centrality recursively. Let $\delta_{st}(v):=\sigma_{st}(v)/\sigma_{st}$ be the ratio of all shortest path between $s$ and $t$ that contain node $v$ compared to all shortest paths between $s$ and $t$. We denote by $\delta_{s*}(v)$ \todoI{how does it relate to BC?}the dependency of a node $s\in V$ on another node $v\in V$, defined by $\delta_{s*}(v):=\sum_{t\in V, t\neq v\neq s}\delta_{st}(v)$. The dependency can be calculated recursively as Brandes  \cite{brandes:2001:fasterBCalgo} stated:

\begin{theorem} (Recursive Betweenness Centrality Dependency, Brandes \cite{brandes:2001:fasterBCalgo} Theorem 6)
$$\delta_{s*}(v)=\sum\limits_{w:v\in P_s(w)}\frac{\sigma_{sv}}{\sigma_{sw}}\cdot (1+\delta_{s*}(w))$$
\label{theo:brandes:recursive}
\end{theorem}

Bader et al. showed in \cite{bader:2007:BCapprox} an adaptive sampling algorithm, which approximates the betweenness centrality with significantly reduced number of single-source shortest path computations for a node with high centrality when using Brandes' recursive algorithm.
\begin{theorem} (Bader et al \cite{bader:2007:BCapprox} Theorem 3).
For $0<\epsilon<\frac{1}{2}$, if the centrality of a vertex $v$ is $n^2/t$ for some constant\todoI{constant ???} $t\geq 1$, then with probability $\geq 1-2\epsilon$ its betweenness centrality can be estimated to within a factor of $1/\epsilon$ with $t\epsilon$ samples of source vertices.
\label{theo:bader:estimation}
\end{theorem}

This algorithm can be adapted efficiently in a distributed setting by using Algorithms\iffull~\ref{alg:DLGcomp}\fi~\textsc{DLGcomp} and\iffull~\ref{alg:DLGagr}\fi~\textsc{DLGagr}.
We define an Algorithm\iffull~\ref{alg:BC_setup_controlling}\fi \textsc{BCsetup} to perform multiple rounds as suggested by Bader in \cite{bader:2007:BCapprox}. In each round, Algorithms\iffull~\ref{alg:DLGcomp}\fi~\textsc{DLGcomp} and\iffull~\ref{alg:DLGagr}\fi~\textsc{DLGagr} are executed with algorithms $\mathcal{A}_{BC}$ and $\mathcal{A}_{BC}'$ as defined \iffull in Algorithm \ref{alg:BC_f} and \ref{alg:BC_g}\fi. With $\mathcal{A}_{BC}$ and $\mathcal{A}_{BC}'$ we state the Algorithms $\mathcal{A}_{agr}$ and $\mathcal{A}_{agr}'$ of our framework specified for betweenness centrality approximation ($BC$) being the function $agr$. 
In contrast to Bader, our algorithm samples not just the betweenness centrality dependency of one node in each round, but of multiple nodes.
Furthermore, Bader's algorithm concentrates on one node $v\in V$ and stops sampling if $BC(v)$ can be approximated within an error $\left(\times,1+\frac{1}{\epsilon}\right)$ with probability $\geq 1-2\epsilon$. Our Algorithm\iffull~\ref{alg:BC_setup_controlling}\fi~\textsc{BCsetup} considers all nodes in the graph and stops after at least $\hat{n}$ nodes have been approximated within a similar multiplicative error $(\times, 1+\frac{1}{\epsilon'})$ with probability $\geq 1-2\epsilon'$ (where $1/n^c<\eps'=\eps-1/n^c<\eps<1/2$).

\subsubsection{Algorithm for Betweenness Centrality}\label{sec:algBC}

In Algorithm\iffull~\ref{alg:BC_setup_controlling}\fi~\textsc{BCsetup}, the idea is to select in multiple rounds multiple sample nodes $s\in S$ and calculate the betweenness centrality dependency $\delta_{s*}(u)$ on all other nodes $u\in V$ for each $s\in S$. The algorithm stops if more than $\hat{n}$ nodes in $V$ with high betweenness centrality are found, $\hat{n}$ is an input parameter. Set $S_i$ is the set of nodes which are sampled in round $i$, all sets $S_i$ are disjoint and form together the set $S=S_1\cup S_2\cup\cdots\cup S_{max}$. 

During the algorithm each node $u\in V$ stores the sum of all the dependencies $\delta_{s*}(u)$ of the nodes $s\in S$ sampled so far in variable $s_{BC}$. Furthermore the number $k_i$ of sample nodes $s\in S_1\cup S_2\cup\cdots\cup S_{i}$ that were sampled so far (in any of rounds $\{1,2,\dots, i\}$), need to be stored to indicate whether the node itself is sampled ($is_{sampled}$).\todoI{why?}

All centralized communication to and from node 1 uses tree $T_1$. In round  $i$ \iffull(Line \ref{line:BCprobstart} to \ref{line:BCprobend})\fi node 1 calculates the probability $p_s$ (according to the proof in Theorem \ref{theo:bcapprox})

with which each node gets sampled. The value of $p_s$ is broadcasted in the network and each node decides if itself is a sample node and reports back to node 1 if so.
\iffull
 In Line \ref{line:BCaccsamples}, if 
\fi
 a node in $T_1$ receives multiple messages in one time slot, it sends the sum of the message value to its parent in $T_1$. The number of sample nodes $|S_i|$ gathered by node 1 is then broadcasted again, this is needed to determine the runtime $|S_i|+D_\omega$ in Algorithms\iffull~\ref{alg:DLGcomp}\fi~\textsc{DLGcomp} and\iffull~\ref{alg:DLGagr}\fi~\textsc{DLGagr} and to maintain the number of sample nodes $k$.

\iffull
\begin{algorithm}[htp]
\begin{algorithmic}[1]
\STATE \textbf{global} $s_{BC}:=0$; \COMMENT{sum of dependency scores $\delta_{v*}(u)$, $v\in V\setminus\{u\}$}
\STATE \textbf{global} $k:=0$; \COMMENT{number of samples so far}
\STATE \textbf{global} $is_{sampled}:=0$;
\STATE $BC(u):=\bot$; \COMMENT{approximated betweenness centrality of $u$, initialized as undefined}
\IF{$u=1$}\label{line:BCestimstart}
        \STATE estimate $D_h'$ and $D_{\omega}'$ by generating a spanning tree $T_1$ and using Fact \ref{fact:ecc-approx-diam};
        \STATE \textbf{broadcast} $D_h'$ and $D_\omega'$ on $T_1$;
        \STATE $n_{BC}:=0$; \COMMENT{number of high BC nodes found so far}
\ELSE
        \STATE \textbf{wait} until value of $D_h'$ and $D_\omega'$ \textbf{received};
\ENDIF\label{line:BCestimend}
\FOR{$T:=log(D_h'),log(D_h')+1,\dots,log(n)$}
        \STATE $t:=2^T$;
        \IF{$u=1$} \label{line:BCprobstart}
                \STATE \textbf{broadcast} $p_s:=\frac{t}{n}$ on $T_1$;
                \STATE \textbf{wait} for responses $r_i$ for $2D_\omega'$ time slots;
                \STATE $|S_i|:=\sum\limits_i r_i$; \COMMENT{number of samples in this round}
                \STATE \textbf{broadcast} $|S_i|$ through $T_1$; 
        \ELSE 
                \STATE  \textbf{receive} value of $p_s$;
                \IF{$is_{sampled}=0$}
                        \STATE \textbf{set} $gets_{sampled}:=1$ \textbf{with probability} $p_s$;
                        \IF{$gets_{sampled}=1$}
                                \STATE \textbf{send} (``1") to node $1$ by $T_1$; \COMMENT{accumulate messages} \label{line:BCaccsamples}
                        \ENDIF
                        \STATE \textbf{receive} value of $|S|$;
                \ENDIF
        \ENDIF \label{line:BCprobend}
        \STATE \textbf{start} Algorithm\iffull~\ref{alg:DLGcomp}\fi~\textsc{DLGcomp} with ($|S_i|$, $D_\omega'$, $\mathcal{A}_{BC}()$);
        \STATE \textbf{start} Algorithm\iffull~\ref{alg:DLGagr}\fi~\textsc{DLGagr} with ($|S_i|$, $D_\omega'$, $\mathcal{A}_{BC}'()$); \label{line:BCstartagg}
        \STATE $k:=k+|S_i|$;
        \IF{$gets_{sampled}=1$}
                \STATE $is_{sampled}:=1$;
                \STATE $gets_{sampled}:=0$;
        \ENDIF
        \IF{$s_{BC}\ge n\cdot \tau_c \textbf{ and } BC(u)=\bot$} \label{line:BCcheck}
                \STATE $BC(u):=\frac{n\cdot s_{BC}}{2k}$; \COMMENT{divided by two, since G is a undirected graph} \label{line:scalebc}
                \STATE \textbf{send} (``reached threshold", 1) to node $1$ by using paths in $T_1$;  \COMMENT{accumulate messages} \label{line:BCaccthresh}
        \ENDIF
        \IF{$u=1$}
                \STATE \textbf{wait to receive} message(s) $M$ for up to $D_\omega'$ time slots;
                \STATE $n_{BC}:=n_{BC}+|M|$;
                \IF {$n_{BC}\ge \hat{n}$} \label{line:BCnhatcheck}
                        \STATE \textbf{broadcast} ``terminate";
                \ENDIF
        \ENDIF
\ENDFOR
\end{algorithmic}
\caption{\textsc{BCsetup}: Distributed Betweenness Centrality Approximation (executed by node $u$) 
\newline \textbf{Input:}  $\tau_c$: threshold criterion, $\hat{n}$: number of nodes which have to reach the threshold $\tau_c$}
\label{alg:BC_setup_controlling}
\vspace*{0.5cm}
\end{algorithm}
\fi

Next, all nodes simultaneously start to execute Algorithm\iffull~\ref{alg:DLGcomp}\fi~\textsc{DLGcomp} with algorithm $\mathcal{A}_{BC}$, whose pseudocode is specified in Algorithm
\iffull
 \ref{alg:BC_f}.
\fi
In round $i$ Algorithm\iffull~\ref{alg:DLGcomp}\fi~\textsc{DLGcomp} computes a DLG $\calL_s$ for every sample node $s\in S_i$. And $\mathcal{A}_{BC}$ ensures that while computing the DLGs all nodes $u\in V$ compute the total number of shortest $u$-$s$-paths ($\tsp$) to all DLG root nodes $s$. Therefore $\mathcal{A}_{BC}$ needs to initialize the array $total_{sp}$ on each node $u\in V$, which stores at $total_{sp}[s]$ the number of shortest path between $u$ and $s$\iffull (Lines \ref{line:fbcinitstart} to \ref{line:fbcinitend})\fi. A sample node $s$ initializes the distance to itself with 1 

.
\iffull
\begin{algorithm}[htb]
\begin{algorithmic}[1]
\STATE \COMMENT{INITIALIZATION} \label{line:fbcinitstart}
\STATE \textbf{global} $total_{sp}:=\{0,0,\dots,0\}$; \COMMENT{$total_{sp}[v]$ := total shortest paths from $v$ to $u$}
\IF{$gets_{sampled}=1$}
        \STATE $total_{sp}[u]:=1$;\label{line:fbcinitend}\newline
\ENDIF

\STATE \COMMENT{COMPUTATION, $f(v, msg_p)$: $v$ root node ID of the message, $msg_p$ message of a parent in $\calL_v$}
\STATE $total_{sp}[v]:=total_{sp}[v]+msg_p$;
\STATE $msg_c[v]:=total_{sp}[v]$; \COMMENT{output}
\end{algorithmic}
\caption{$\mathcal{A}_{BC}()$}
\label{alg:BC_f}
\vspace*{0.5cm}
\end{algorithm}
\fi

During the execution of Algorithm\iffull~\ref{alg:DLGcomp}\fi~\textsc{DLGcomp} in node $u,$ algorithm $\mathcal{A}_{BC}$ processes messages ($msg_p$) from parents $u_h\in\{u_1,\dots,u_{deg(u)}\}$ in DLG $\calL_s$, which contains the value $\tsp$ from  $u_h$ to $s$.
Algorithm $\mathcal{A}_{BC}$ then updates the corresponding entry $total_{sp}[s]$ by adding the value $\tsp$ received from parent $u_h$. Entry $total_{sp}[s]$ stores the correct value after all messages of parents in $\calL_s$ have arrived. Node $u$ then forwards $total_{sp}[s]$ \iffull during the execution of Algorithm~\ref{alg:DLGcomp}~\textsc{DLGcomp} in Line \ref{line:sendreceivedist}\fi to all its children in $\calL_s$ (encoded in message $msg_c[s]$).

In Algorithm\iffull~\ref{alg:DLGagr}\fi~\textsc{DLGagr} each node $u\in V$ adds the dependencies $\delta_{s*}(u)$ of all sampled nodes $s\in S_i$ on itself to the sum $s_{BC}$. Each time a child node $w$ in a DLG $\calL_s$ sends a message to $u$, algorithm $\mathcal{A}_{BC}'$, whose pseudocode is specified in Algorithm
\iffull
 \ref{alg:BC_g}
\fi
, gets executed. The message ($msg_c=(\delta_c,\tsp_c)$) of child $w$ in $\calL_s$ contains the dependency $\delta_{s*}(w)$ and the $\tsp$ from $w$ to $s$. The dependencies $\delta_{s*}(u)$\todoI{these are approx, should they be called $\tilde{\delta}_{s*}(u)$?} are then calculated recursively with the formula in Theorem \ref{theo:brandes:recursive} and stored in $\delta_{BC}[s]$.
Entry $\delta_{BC}[s]$ stores the correct value after all messages of children in $\calL_s$ have arrived. Node $u$ then forwards $(\delta_{BC}[s], total_{sp}[s])$ \iffull during the execution of Algorithm~\ref{alg:DLGagr}~\textsc{DLGagr} in Line \ref{line:sendreceiveagg}\fi to all its parents in $\calL_s$ (in message $msg_p[s]$). In a leaf node $l$ of DLG $\calL_s$, Algorithm $\mathcal{A}_{BC}'$ never gets executed and therefore forwards the initialization value of $\delta_{BC}[s]=0$ to its parents in $\calL_s$.

\iffull
\begin{algorithm}[htb]
\begin{algorithmic}[1]
\STATE \COMMENT{INITIALIZATION} 
\STATE \textbf{global} $\delta_{BC}:=\{0,\dots,0\}$; \COMMENT{$\delta_{BC}[s]$ := dependency of the source vertex s on u.} \newline

\STATE \COMMENT{COMPUTATION, $g(v, msg_c)$: $v$ root node ID of the message, $msg_c$ message of a child in $\calL_v$}
\STATE $(\delta_c, tsp_c):=msg_c$;
\STATE $\delta:=(total_{sp}[v_i]/tsp_c)\cdot(1+\delta_c)$;\label{line:dep1bc}
\STATE $\delta_{BC}[v_i]:=\delta_{BC}[v_i]+\delta$;\label{line:depbc}
\STATE $s_{BC} := s_{BC}+\delta$; \label{line:sumsbc}
\STATE $msg_p[v]=(\delta_{BC}[v], total_{sp}[v])$; \COMMENT{output}
\end{algorithmic}
\caption{$\mathcal{A}_{BC}'()$}
\label{alg:BC_g}
\vspace*{0.5cm}
\end{algorithm}
\fi

After the execution of Algorithm\iffull~\ref{alg:DLGagr}\fi~\textsc{DLGagr} each node checks \iffull in Line \ref{line:BCcheck}\fi if $s_{BC}$ passed the threshold $n\cdot \tau_c$ and reports back to node 1 if this is the case. Node $1$ adds the number of passed threshold to $n_{BC}$, a variable that indicates how many nodes passed the threshold so far. In \cite{bader:2007:BCapprox} it is stated that $\tau_c=5$ yields a good trade off between computation (time) and approximation quality. To prevent congestion when multiple nodes reached the threshold and therefore multiple messages are sent back to node 1, we use the sum-aggregation of our framework as seen in Example
\iffull
 \ref{ex:sumval}.
\fi

If not enough nodes reached the threshold so far, i.e. $n_{BC}<\hat{n}$, another round with additional sample nodes gets initiated by node 1.

\begin{remark}
Note, that in an undirected graph the calculated value of $BC(u)$ has to be divided by 2, since for every node pair $(s,t)\in V\times V$ the betweenness dependency of the shortest $s$-$t$-paths on $u$ is computed twice, once as a part of $\delta_{s*}(u)$ and once as a part of $\delta_{t*}(u)$.
\end{remark}

\begin{lemma}
Algorithm\iffull~\ref{alg:BC_setup_controlling}\fi~\textsc{BCsetup} computes $\delta_{s*}(u)$ for each node $u\in V$ and all samples $s\in S_i$ in time $\BO(|S_i|+D_\omega)$ in round $i$, with a multiplicative approximation guarantee in the range of $[1-\frac{1}{n^c},1+\frac{1}{n^c}]$. 
\label{lemma:deproundapprox}
\end{lemma}

\iffull
\begin{proof}
\textbf{Runtime: }
The computation of broadcasting and aggregating $p_s$ and $|S_i|$ \iffull in Lines \ref{line:BCprobstart} to \ref{line:BCprobend} \fi takes $\BO(D_\omega)$ time slots. The execution of Algorithms\iffull~\ref{alg:DLGcomp}\fi~\textsc{DLGcomp} and\iffull~\ref{alg:DLGagr}\fi~\textsc{DLGagr} is bounded when using our framework by $|S_i|+D_\omega'$, see Theorem \ref{theo:DLGruntime}. Therefore the runtime of Algorithm\iffull~\ref{alg:BC_setup_controlling}\fi~\textsc{BCsetup} is bound by $\BO(|S_i|+D_\omega)$.

\textbf{Correctness: }
As defined in Theorem \ref{theo:brandes:recursive}, to calculate $\delta_{s*}(u),   $ it suffices that a node $u$ needs to know the number of shortest $u$-$s$-paths $\sigma_{us}$ and for every child $u_h$ of $u$ in $\calL_s$ the number of $u_h$-$s$-paths $\sigma_{u_hs}$ and the child's $\delta_{s*}(u_h)$. The number of shortest $u$-$s$-paths of a node $u$ can be obtained by adding the number of shortest paths of all parents in $\calL_s$ together, which is done in $\mathcal{A}_{BC}$ along the execution of Algorithm\iffull~\ref{alg:DLGcomp}\fi~\textsc{DLGcomp}. A leaf $u_l$ in $\calL_s$ has a betweenness dependency of zero per definition. Therefore Algorithm\iffull~\ref{alg:DLGagr}\fi~\textsc{DLGagr} and $\mathcal{A}_{BC}'$ can start aggregating the missing parameters in the leafs of $\calL_s$. At the end of the execution of Algorithm\iffull~\ref{alg:DLGagr}\fi~\textsc{DLGagr} in \iffull Line \ref{line:BCstartagg} of \fi Algorithm\iffull~\ref{alg:BC_setup_controlling}\fi~\textsc{BCsetup}, each node $u\in V$ has the information to calculate $\delta_{s*}(u)$.
\todoI{@stephan: this is the only reference to the whole shortest path proof}

While executing Algorithm\iffull~\ref{alg:DLGcomp}\fi~\textsc{DLGcomp} to compute a $DLG$ $\calL_s$, a node $u$ sends the total number of shortest paths ($\tsp$) between $s$ and $u$ to its children. 
As the graph in Figure \ref{abb:example2n} shows, this number can be as large as $\BO(2^n)$ and therefore $\BO(n)$ bits are needed to transmit the exact number. In order to use not more than $\BO(\log n)$ bits per message, we approximated $\tsp$ by using a $(4+c)\log n$ bits: $(3+c)\log(n)$ bit for the most significant bits $b$  of $\tsp$ and $\log(n)$ bits to encode the exponent $E$ of $tsp$, that is $tsp=b\cdot 2^E$. Therefore the approximation of $\tsp$ has a multiplicative approximation guarantee in the range of $[1-\delta,1+\delta]$ for the betweenness centrality with $|\delta|\leq\frac{1}{n^c}$. We start with the proof for the case of an unweighted graph. And state later in Lemma \ref{lemma:shortesp:weighted} that it holds too for weighted graphs.

\begin{figure}[h]
\begin{center}
\includegraphics{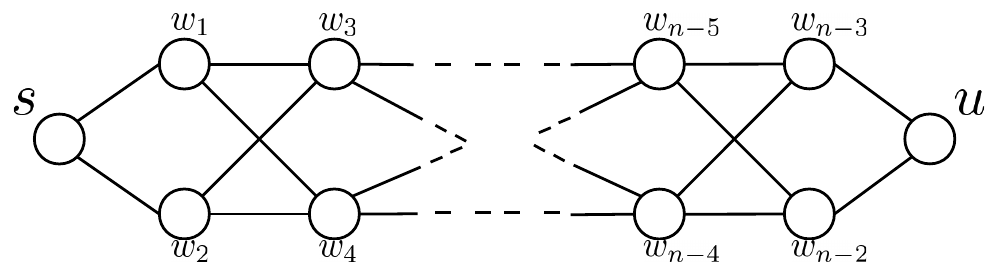}
\caption{Graph with $2^{\frac{n-2}{2}}$ shortest $s$-$u$-paths}\label{abb:example2n}
\end{center}
\end{figure}

Let $\tsp_k(u)$ be the exact and $\tspapp_k(u)$ the approximated number of total shortest paths from a sample node $s$ to $u$, where $k$ is the hop distance of $u$ to $s$.

In Equation \eqref{eq:tsp0} the number of shortest paths from a node $u$ to itself is set to $1$. This is done for simplicity--note that the actual number is $0$--as it ensures that the number of shortest paths from $u$ to ant neighbor is (at least) $1$. And since there is only one shortest path between two nodes which are one hop apart, consisting of the edge between them, Equation \eqref{eq:tsp1} follows.
The exact computation of $\tsp(u)$, which is dependent of the values received from parents of $u$, is shown in Equation \eqref{eq:tsprec}. It is the sum of all shortest path from $s$ to all parents of $u$.

\begin{align}
&\tsp\nolimits_0(s)=1 \label{eq:tsp0}\\
&\tsp\nolimits_1(u)=1 \label{eq:tsp1}\\
&\tsp\nolimits_k(u)=\sum_{v\in P(u)}\tsp\nolimits_{k-1}(v)
\label{eq:tsprec}
\end{align}

We proof by induction a lower bound for the approximation error at a node $u$ dependent of the hop distance $k$ from $u$ to the root node $s$.

\begin{align}
\textbf{Claim: }&& 
\tspapp\nolimits_k(u)&\geq(1-\varepsilon)^{k-1}\cdot\tsp\nolimits_k(u)
\label{eq:tspclaim}  \allowdisplaybreaks[4] \\
\textbf{Base: }&& 
\tspapp\nolimits_1(u)&\geq(1-\varepsilon)^0\cdot\tsp\nolimits_0(u)=1\cdot 1=1
\label{eq:tspbase}  \allowdisplaybreaks[4] \\
\textbf{Step: }&& 
\tspapp\nolimits_{k+1}(u) &\geq (1-\varepsilon)\cdot\left(\sum_{v\in P(u)}\tspapp\nolimits_{k}(v)\right)\nonumber\\
&& &\geq (1-\varepsilon)\cdot\left(\sum_{v\in P(u)}(1-\varepsilon)^{k-1}\tsp\nolimits_{k}(v)\right)\nonumber\\
&& &= (1-\varepsilon)^{k}\cdot\left(\sum_{v\in P(u)}\tsp\nolimits_{k}(v)\right)\nonumber\\
&& &= (1-\varepsilon)^{k}\cdot\tsp\nolimits_{k+1}(u) & \square
\label{eq:tspstep}
\end{align}

From Equation \eqref{eq:tspclaim} we derive a maximum lower error bound independent of $k$ for all nodes $u\in V$, as shown in Equation \eqref{eq:tsprecappend}, using the fact that $k\leq D_h$.

\begin{align}
\tspapp\nolimits_k(u)\geq(1-\epsilon)^{D_h}\cdot\tsp\nolimits_k(u)\quad\forall 1\leq k\leq D_h
\label{eq:tsprecappend}
\end{align}

Now let $\dep(u)$ be the exact and $\depapp(u)$ the approximated betweenness dependency of a sample node $s$ on $u$, where $k$ is the hop distance of $u$ to $s$. Let us assume $\tsp$ got approximated at every node with the maximum error $(1-\epsilon_m)=(1-\epsilon)^{D_h}$.
All nodes without any children in $DLG_s$ (leafs), e.g. nodes at maximum hop distance of $s$, have no betweenness dependency of $s$, since there are no shortest paths going \textit{through} them, this proves Equation \eqref{eq:bcrec0}.
Then the exact recursive computation of $\dep$ is shown in Equation \eqref{eq:bcrec}.

\begin{align}
&\dep\nolimits_{D_h}(u)=\dep\nolimits_{leaf}(u)=0 \label{eq:bcrec0}\\
&\dep\nolimits_k(u)=\sum_{w:u\in P(w)}\frac{\tsp\nolimits_{k}(u)}{\tsp\nolimits_{k+1}(w)}\cdot\left(1+\dep\nolimits_{k+1}(w)\right) \label{eq:bcrec}
\end{align}

Now we derive error bounds by evaluating the term $\frac{\tspapp_i}{\tspapp_{i+1}}$. The maximum errors occur at $\frac{\tspapp_i}{\tspapp_{i+1}}=\frac{\tsp_i\cdot(1-\epsilon_m)}{\tsp_{i+1}}$(i) and at $\frac{\tspapp_i}{\tspapp_{i+1}}=\frac{\tsp_i}{\tsp_{i+1}\cdot(1-\epsilon_m)}$(ii).
With (i) and (ii) we show by induction an upper and lower bound for the multiplicative error of $\depapp_k$ in Equations \eqref{eq:depclaimi} to \eqref{eq:depstepii}.
\newline\newline
Proof by Induction on the lower approximation error bound of $\depapp$:
\begin{align}
\textbf{Claim: }&&
\depapp\nolimits_k(u)&\geq(1-\epsilon_m)^{D_h-k}\cdot\dep\nolimits_k(u) \label{eq:depclaimi}  \allowdisplaybreaks[4] \\
\textbf{Base: }&&
\depapp\nolimits_{D_h-1}(u)&=\sum_{w:u\in P(w)}\frac{\tspapp\nolimits_{D_h-1}(u)}{\tspapp\nolimits_{D_h}(w)}\cdot\left(1+0\right)\nonumber\\
&& &\geq\sum_{w:u\in P(w)}\frac{\tsp\nolimits_{D_h-1}(u)\cdot(1-\varepsilon_m)}{\tsp\nolimits_{D_h}(w)}=(1-\varepsilon_m)\cdot\dep\nolimits_{D_h-1}(u) \label{eq:depbasei}  \allowdisplaybreaks[4] \\
\textbf{Step: }&&
\depapp\nolimits_{k-1}(u)&=\sum_{w:u\in P(w)}\frac{\tspapp\nolimits_{k-1}(u)}{\tspapp\nolimits_{k}(w)}\cdot\left(1+\depapp\nolimits_{k}(w)\right)\nonumber\\
&& &\geq\sum_{w:u\in P(w)}\frac{\tsp\nolimits_{k-1}(u)\cdot(1-\epsilon_m)}{\tsp\nolimits_{k}(w)}\cdot\left(1+(1-\epsilon_m)^{D_h-k}\cdot\dep\nolimits_k(w)\right)\nonumber\\
&& &\geq(1-\epsilon_m)^{D_h-(k-1)}\cdot\sum_{w:u\in P(w)}\frac{\tsp\nolimits_{k-1}(u)}{\tsp\nolimits_{k}(w)}\cdot\left(1+\dep\nolimits_k(w)\right)\nonumber\\
&& &=(1-\epsilon_m)^{D_h-(k-1)}\cdot\dep\nolimits_{k-1} & \square \label{eq:depstepi}
\end{align}
\newline
Proof by Induction on the upper approximation error bound of $\depapp$:
\begin{align}
\textbf{Claim: }&&
\depapp\nolimits_k(u)&\leq(1-\epsilon_m)^{k-D_h}\cdot\dep\nolimits_k(u)\label{eq:depclaimii} \allowdisplaybreaks[4] \\
\textbf{Base: }&&
\depapp\nolimits_{D_h-1}(u)&=\sum_{w:u\in P(w)}\frac{\tspapp\nolimits_{D_h-1}(u)}{\tspapp\nolimits_{D_h}(w)}\cdot\left(1+0\right)\nonumber\\
&& &\leq\sum_{w:u\in P(w)}\frac{\tsp\nolimits_{D_h-1}(u)}{\tsp\nolimits_{D_h}(w)\cdot(1-\varepsilon_m)}=(1-\varepsilon_m)^{-1}\cdot\dep\nolimits_{D_h-1}(u)\label{eq:depbaseii}  \allowdisplaybreaks[4]\\
\textbf{Step: }&&
\depapp\nolimits_{k-1}(u)&=\sum_{w:u\in P(w)}\frac{\tspapp\nolimits_{k-1}(u)}{\tspapp\nolimits_{k}(w)}\cdot\left(1+\depapp\nolimits_{k}(w)\right)\nonumber\\
&& &\leq\sum_{w:u\in P(w)}\frac{\tsp\nolimits_{k-1}(u)}{\tsp\nolimits_{k}(w)\cdot(1-\epsilon_m)}\cdot\left(1+(1-\epsilon_m)^{k-D_h}\cdot\dep\nolimits_k(w)\right)\nonumber\\
&& &\leq(1-\epsilon_m)^{(k-1)-D_h}\cdot\sum_{w:u\in P(w)}\frac{\tsp\nolimits_{k-1}(u)}{\tsp\nolimits_{k}(w)}\cdot\left(1+\dep\nolimits_k(w)\right)\nonumber\\
&& &=(1-\epsilon_m)^{(k-1)-D_h}\cdot\dep\nolimits_{k-1} & \square
\label{eq:depstepii}
\end{align}

From the upper and lower approximation error bounds in a node $u$ with a certain hop distance $k$ to the root node, we can derive in Equations \eqref{eq:depclaimi} and \eqref{eq:depclaimii} a maximum error range for any node at any distance independent of $k$, as shown in Equation \eqref{eq:depbounds}.
\todoI{$\eps$ is later $n^{-c-3}$}
\begin{align}
(1-\epsilon_m)^{D_h-1}\cdot\dep\nolimits_k(u)\leq\depapp\nolimits_k(u) &\leq (1-\epsilon_m)^{1-D_h}\cdot\dep\nolimits_k(u) \nonumber\\ 
(1-\epsilon)^{D_h^2}\cdot\dep\nolimits_k(u)\leq\depapp\nolimits_k(u) &\leq(1-\epsilon)^{-D_h^2}\cdot\dep\nolimits_k(u)\qquad\forall 1\leq k\leq D_h
\label{eq:depbounds}
\end{align}

By using the fact, that $D_h\leq n$ we derive in Equations \eqref{eq:epsilondeltamin} to \eqref{eq:delta} a tight bound for $|\delta|\leq\frac{1}{n}$.

\begin{equation}
(1+\delta)\geq(1-\epsilon)^{D_h^2}\overbrace{\geq}^{\epsilon=\frac{1}{n^{3+c}}\leq\frac{1}{n^c\cdot D_h^2}}\left(1-\frac{1}{n^c\cdot D_h^2}\right)^{D_h^2}\approx 1-\frac{1}{n^c}
\label{eq:epsilondeltamin}
\end{equation}

\begin{equation}
(1+\delta)\leq(1-\epsilon)^{-D_h^2}\overbrace{\leq}^{\epsilon=\frac{1}{n^{3+c}}\leq\frac{1}{n^c\cdot D_h^2}}\left(1-\frac{1}{n^c\cdot D_h^2}\right)^{-D_h^2}\approx 1+\frac{1}{n^c}
\label{eq:epsilondeltamax}
\end{equation}

\begin{equation}
-\frac{1}{n^c}\leq\delta\leq\frac{1}{n^c} \qquad \square
\label{eq:delta}
\end{equation}

\begin{corollary}
For betweenness centrality on unweighted graphs we obtain an approximation of the total shortest paths which is guaranteed to be in the range of a $[1-\delta,1+\delta]$ multiplicative factor, where $|\delta|\leq\frac{1}{n^c}$, when assuming that the total shortest path error is multiplicatively bounded to be smaller than $\frac{1}{n^{3+c}}$.
\label{cor:tsperrunw}
\end{corollary}
\todoI{
\begin{proof}
The Error Bound \eqref{eq:delta} holds for the dependencies calculated in Algorithm \ref{alg:BC_g}, Line \ref{line:dep1bc} and \ref{line:depbc}. By adding multiple dependencies, in Line \ref{line:sumsbc} of the same algorithm, the error bounds do not change. Neither do they change in Algorithm\iffull~\ref{alg:BC_setup_controlling}\fi~\textsc{BCsetup}, Line \ref{line:scalebc}, by scaling up the dependencies of $S$ on $u$ to the approximated betweenness centrality $BC(u)$. This is true, as in sumations and when scaling equations, the multiplicative error can simply \todo{be factored out -- what does this mean?}. \blue{$BC_{unscaled}=OPT*\varepsilon$ and $BC_{scaled}=\frac{k*BC_{unscaled}}{n}$, it holds $BC_{scaled}=\frac{k*OPT*\varepsilon}{n}=\frac{k*OPT}{n}*\varepsilon=BC_{OPT}*\varepsilon$, d.h. $ 
\varepsilon$ aendert sich nicht}
\label{lemma:tsperrw}
\end{proof}
}

\begin{lemma}
The approximation guarantee of a factor within the range $[1-\delta, 1+\delta]$ with $|\delta|\leq\frac{1}{n^c}$ of the betweenness centrality for unweighted graphs holds for weighted graphs, too.
\label{lemma:shortesp:weighted}
\end{lemma}
To prove this, we need to introduce the shortest path diameter $D_{sp}$.
\begin{definition}[Shortest path diameter]
$D_{sp}(G):=\max_{u,v\in V}\{|P||P \text{is a shortest weighted u-v-path}\}$ of a graph G is the maximum number of hops of a shortest weighted path between any two nodes of the graph.\\
Note that $1\leq D_h\leq D_{sp}\leq n$.
\label{def:diamsp}
\end{definition}

\begin{proof}
The maximum hop distance between any leaf $l$ in a DLG and its root node $r$ can be longer than $D_h$. Because a shortest weighted path is not bounded by $D_h$ hops, but by $D_{sp}$ as defined in Definition \ref{def:diamsp}. That changes the maximum exponents in Equations \eqref{eq:tsprecappend}, \eqref{eq:depclaimi} and \eqref{eq:depclaimii} from $D_h$ to $D_{sp}$. Since $D_{sp}$ is smaller than $n$ too, the error bounds in Equations \eqref{eq:epsilondeltamin} and \eqref{eq:epsilondeltamax} do not change.
\end{proof}

\end{proof}

\begin{proof}[Proof of Theorem \ref{theo:bcapprox}]
\textbf{Runtime: }
The computation of $D_h$ and $D_\omega$ \iffull in Lines \ref{line:BCestimstart} to \ref{line:BCestimend}\fi takes $\BO(D_\omega)$ time slots.
Each round takes $\BO(|S_i|+D_\omega)$ time as shown in Lemma \ref{lemma:deproundapprox}. And for sampling $|S|$ nodes we need $\log\frac{|S|}{D_h}$ rounds, since we start with ${D_h}$ sample nodes and double this value in every round. This leads to a total time complexity of $\BO\left(|S|+D_\omega\log\frac{|S|}{D_h}\right)$.

\textbf{Correctness: }Due to Theorem \ref{theo:bader:estimation}, the betweenness centrality $BC(v)$ of a node $v$ can be estimated with probability $\geq 1-2\epsilon$ by sampling the dependency $\delta_{s*}(v)$ of $t\epsilon=|S|$ root nodes $s\in S$, if $BC(v)\geq\frac{n^2}{t}$.
The approximation error given in Theorem \ref{theo:bader:estimation} is $\left(\times,1+\frac{1}{\epsilon}\right)$.

As shown in Lemma \ref{lemma:deproundapprox}, Algorithm\iffull~\ref{alg:BC_setup_controlling}\fi~\textsc{BCsetup} approximates those dependencies $\delta_{s*}(v)$ within a with a multiplicative approximation guarantee in the range of $[1-\frac{1}{n^c},1+\frac{1}{n^c}]$. 
\iffull In Line \ref{line:BCcheck} and \ref{line:BCnhatcheck}, \fi
Algorithm\iffull~\ref{alg:BC_setup_controlling}\fi~\textsc{BCsetup} ensures that $s_{BC}\ge n\cdot \tau_c$ root nodes have been sampled to approximate the value of $BC(v)$ of a node $v$ such that the betweenness centralities of at least $\hat{n}$ nodes are approximated close enough before terminating.
Combining both approximation errors and noting that $\frac{1}{n^c}<\frac{1}{\epsilon}$, leads to a total error of $\left(\times,1+\frac{1}{\epsilon'}\right)$.
\end{proof}

\begin{corollary}
Let $BC(v)$ be the betweenness centrality of node $v$. Then $\BO\left(\frac{n^2}{BC(v)} + D_w\log \frac{n}{BC(v)\cdot D_h}\right)$ rounds of communication are necessary to compute an $\left(\times,1+\frac{1}{\epsilon'}\right)$-approximation of $BC(v)$. Furthermore, in the same time our Algorithm \iffull~\ref{alg:BC_setup_controlling}\fi~\textsc{BCsetup} finds  w.h.p all other nodes with higher betweenness centrality than $v$.
\end{corollary}
\fi

\subsection{Closeness Centrality}
\label{sec:CC}

\todoI{It is mentioned in \cite{aggarwal2010survey} that a modification of \cite{rattigan2007graph} chooses landmarks using local closeness centrality - can we do that, too?}\blueI{I dont know. Didnt read it}

Eppstein and Wang showed in \cite{eppstein:2001:fastCCapprox} a fast approximation algorithm to derive the closeness centrality of a node by sampling only a few nodes $S$ and computing shortest paths from $S$ to all other nodes ($S$-SP) instead of calculating all pairs shortest paths (APSP).
These ideas can be implemented efficiently in a distributed setting by using Algorithm\iffull~\ref{alg:DLGcomp}\fi~\textsc{DLGcomp} of our proposed framework.
\iffull
The Algorithm \ref{alg:CC} first computes $T_1$ with a breath-first-search algorithm starting in node $1$, then the value of $\epsilon$ gets broadcasted through $T_1$ in the network.
When a node has received $\epsilon$, it joins $S$ as a sample node with probability $\frac{log(n)}{n\cdot \epsilon^2}$. The number of nodes joining set $S$ are collected and broadcasted in Lines \ref{line:CCjoinSstart} to \ref{line:CCjoinSend}. In Line \ref{line:CCaccmessages}, if a node $u\in V$ receives multiple messages in one time slot, then in the next time slot node $u$ sends the sum of the values of all messages to its parent in $T_1$.
Now, in line \ref{line:CCexAlg}, each node executes Algorithm\iffull~\ref{alg:DLGcomp}\fi~\textsc{DLGcomp}. No algorithm $\mathcal{A}_{agr}$ or $\mathcal{A}_{agr}'$ 
specified for Closeness Centrality is needed, as Algorithm\iffull~\ref{alg:DLGcomp}\fi~\textsc{DLGcomp} already computes the required distances.
After Algorithm\iffull~\ref{alg:DLGcomp}\fi~\textsc{DLGcomp} is executed, every node knows its distances $\omega[r_i]$ to all sample nodes $r_i\in S$ and therefore the number of samples $|S|$. Using those distances, every node $u$ can approximate its closeness centrality $CC_u$. \blueI{note: no aggregation / Algo2 needed.}

\begin{algorithm}[htb]
\begin{algorithmic}[1]
\IF{$u=1$}
        \STATE \textbf{compute} tree $T_{1}$ and \textbf{estimate} $D_\omega'$ by using Fact \ref{fact:ecc-approx-diam};
        \STATE \textbf{broadcast} values of $\epsilon$ and $D_\omega'$ on $T_{1}$;
\ELSE
        \STATE \textbf{wait} until values of $\epsilon$ and $D_\omega'$ \textbf{received};
\ENDIF
\STATE $u$ \textbf{joins} $S$ with probability $\frac{log(n)}{n\cdot \epsilon^2}$;
\IF{$u\in S$} \label{line:CCjoinSstart}
        \STATE \textbf{send} (``1'') to node $1$ by $T_1$; \COMMENT{accumulate messages}\label{line:CCaccmessages}
\ENDIF
\IF{$u=1$}
        \STATE \textbf{wait} for responses $r_i$;
        \STATE $|S|:=\sum_i r_i$;
        \STATE \textbf{broadcast} $|S|$ through $T_1$;
\ELSE
        \STATE \textbf{receive} value of $|S|$;
\ENDIF \label{line:CCjoinSend}
\STATE \textbf{execute} Algorithm\iffull~\ref{alg:DLGcomp}\fi~\textsc{DLGcomp} with $(|S|,D_\omega', \bot)$; \label{line:CCexAlg}
\STATE $CC_u:=1/\sum_{i=1}^{|S|}\frac{n\cdot \omega[r_i]}{|S|\cdot (n-1)}$ //** $\omega[r_i]$; is the distance of $u$ to sample node $r_i$
\end{algorithmic}
\caption{$\mathcal{A}_{CC}$ Approximation (executed by node $u$) \newline Input: $\epsilon$ }\label{alg:CC}
\vspace*{0.5cm}
\end{algorithm}

\begin{lemma}
Algorithm \ref{alg:CC} approximates closeness centrality of all nodes with an inverse additive error of $\epsilon D_h$ in time $\BO\left(\frac{log(n)}{\epsilon^2}+D_h\right)$ in unweighted graphs.
\label{lemma:ccunweighted}
\end{lemma}

\begin{proof}
As stated by Epstein and Wang in \cite{eppstein:2001:fastCCapprox}, by using $k=\BO\left(\frac{log(n)}{\epsilon^2}\right)$ samples, we can approximate the closeness centrality w.h.p. with an inverse additive error of $\epsilon D_h$. As shown in Theorem \ref{theo:DLGruntime}, Algorithm\iffull~\ref{alg:DLGcomp}\fi~\textsc{DLGcomp} takes $\BO(|S|+D_\omega)$ time to complete. Computing, estimating and/or broadcasting $T_1$, $D_\omega'$ and $\epsilon$ takes in each case $\BO(D_\omega)$ time. Algorithm \ref{alg:CC} therefore takes $\BO\left(\frac{log(n)}{\epsilon^2}+D_h\right)$ time for unweighed graphs, since $D_h=D_\omega$ in unweighted graphs.
\end{proof}

\begin{proof}[Proof of Theorem \ref{theo:ccapprox}]
In the proof of Epstein and Wang in \cite{eppstein:2001:fastCCapprox}, all $D_h$ can be substituted by $D_\omega$ and the proof is still valid. Hence, by using $k=\BO\left(\frac{log(n)}{\epsilon^2}\right)$ samples we can approximate the closeness centrality w.h.p. with an inverse additive error of $\epsilon D_\omega$.

As shown in Lemma \ref{lemma:ccunweighted} Algorithm \ref{alg:CC} takes $\BO(|S|+D_\omega)$ time to complete. Thus, by sampling $k=\BO\left(\frac{log(n)}{\epsilon^2}\right)$ root nodes, Algorithm \ref{alg:CC} needs  time $\BO\left(\frac{log(n)}{\epsilon^2}+D_\omega\right)$.
\end{proof} 
\fi

\iffull
\subsection{Minimum Routing Cost Tree}
\label{app:MRCT}

The S-MRCT-Algorithm in \cite{hochuli:holzer:MRCST} is executed in a similar fashion as our proposed Algorithms\iffull~\ref{alg:DLGcomp}\fi~\textsc{DLGcomp} and\iffull~\ref{alg:DLGagr}\fi~\textsc{DLGagr}, as we generalized the aggregation concept of the S-MRCT-Algorithm \cite{hochuli:holzer:MRCST}.
Therefore the S-MRCT-Algorithm as stated in \cite{hochuli:holzer:MRCST} can be simulated with the Tree Variation of our proposed framework, as stated in Algorithms\iffull~\ref{alg:TreeComputing}\fi~\textsc{TreeComputing} and\iffull~\ref{alg:TreeAggregating}\fi~\textsc{TreeAggregating} .

The S-MRCT-Algorithm in \cite{hochuli:holzer:MRCST} computes a ($\times$,2)-approximation of the S-Minimum Routing Cost Spanning Tree (S-MRCT) problem. It relies on the recursive computation of  routing costs as stated in Lemma \ref{lemma:recRC}. The number of nodes of the set $S$ which are in the subtree of a child $c_i$ is denoted as $|Z_{\left(u,c_i\right)}^2\left(T\right)|=|\{v\in S:v\in T|_{c_i}\}|$ and $|Z_{\left(u,c_i\right)}^1\left(T\right)|=|\{v\in S:v\notin T|_{c_i}\}|$ \todoI{update} denotes the number of all other nodes in set $S$.

\begin{lemma}[Lemma 5.6 in \cite{hochuli:holzer:MRCST}]\label{lemma:recRC}
\[rc_S\left(T,u\right)  = \sum_{i=1}^{deg\left(u\right)-1} rc_S\left(T,c_i\right) + \sum_{i=1}^{deg\left(u\right)-1} \omega\left(u,c_i\right)\cdot |Z_{\left(u,c_i\right)}^1\left(T\right)|\cdot|Z_{\left(u,c_i\right)}^2\left(T\right)|\]
\end{lemma}

Lemma \ref{lemma:recRC} in words: the routing cost $rc_S(T,u)$ of a subtree $T|_{u}$ rooted in $u$ of tree $T$ is the sum of two sums. First, the sum of the routing cost $rc_S(T,c_i)$ of all subtrees $T|_{c_i}$ rooted in the children $c_i$ of $u$. And second, the sum of the routing cost off all edges $(u,c_i)$ from $u$ to all its children $c_i$. The product $|Z_{\left(u,c_i\right)}^1\left(T\right)|\cdot|Z_{\left(u,c_i\right)}^2\left(T\right)|$ denotes the number of $v$-$w$-paths in $T$ which include the edge $(u,c_i)$, for $v,w\in S$. Note that we set the routing cost $rc_S(T,u_l)$ to be $0$ for a node $u_l$ which is a leaf in $T$. And $rc_S(T,r):=RC_S(T)$ for the root node $r$ of tree $T$\todoI{Notation anpassen}.

The idea of the S-MRCT-Algorithm in \cite{hochuli:holzer:MRCST} is, in a first step to compute in parallel $|S|$ trees $T_r$ rooted in $r$ for each node $r\in S$. And in a second step to compute the routing cost $RC_S(T_r)$ for each tree $T_r$ in parallel. Therefore the algorithm aggregates the routing costs $rc_S(T_r,u)$ of the subtrees $T_r|_u$ in tree $T_r$ in a bottom-up fashion with the formula stated in Lemma \ref{lemma:recRC}. Finally the algorithm selects the tree $T_r$ with the minimum routing cost $RC_S(T_r)$ and broadcasts the result in the network.

We use Algorithm\iffull~\ref{alg:MRCT}\fi~\textsc{MRCT}  and Algorithm $\mathcal{A}_{MRCT}'$ (whose pseudocode is stated in Algorithm \ref{alg:MRCT_g}) to transfer the S-MRCT-Algorithm from \cite{hochuli:holzer:MRCST} into our framework. The subalgorithm $\mathcal{A}_{MRCT}'$ denotes the algorithm $\mathcal{A}_{agr}'$ of our framework, specified for the S-MRCT approximation.
Since we are only interested in Trees and not in DLGs, the Tree Variation involving Algorithms\iffull~\ref{alg:TreeComputing}\fi~\textsc{TreeComputing}  and\iffull~\ref{alg:TreeAggregating}\fi~\textsc{TreeAggregating}  is used. 
Assume each node $u\in V$ knows weather it is in set $S$ and $u$ knows $|S|$. Further the weight of an edge denotes the routing cost for that edge.

Algorithm\iffull~\ref{alg:MRCT}\fi~\textsc{MRCT}  starts with estimating and distributing $D_\omega'$ in the graph (Line \ref{line:MRCTstart} to \ref{line:MRCTendestim}). This is needed to determine the runtime of Algorithms\iffull~\ref{alg:TreeComputing}\fi~\textsc{TreeComputing}  and\iffull~\ref{alg:TreeAggregating}\fi~\textsc{TreeAggregating} . Algorithm\iffull~\ref{alg:TreeComputing}\fi~\textsc{TreeComputing}  executed in Line \ref{line:MRCTalgdist} is only needed to compute the trees $T_r$ for Algorithm\iffull~\ref{alg:TreeAggregating}\fi~\textsc{TreeAggregating}  for each node $r\in S$. Therefore we do not need algorithm $\mathcal{A}_{MRCT}$ for the execution. During the aggregating in Line \ref{line:MRCTalgagg}, algorithm $\mathcal{A}_{MRCT}'$ executed inside Algorithm\iffull~\ref{alg:TreeAggregating}\fi~\textsc{TreeAggregating}  computes the routing costs for the trees $T_r$ by applying the recursive formula stated above. Algorithm $\mathcal{A}_{MRCT}'$ therefore uses the recursive relation of the routing cost stated in Lemma \ref{lemma:recRC}.

In $\mathcal{A}_{MRCT}'$ Line \ref{line:gmrctinitstart} to \ref{line:gmrctinitend}, each node $u\in V$ is initialized with zero routing cost $rc_S[r]=0$ (per root node $r\in S$) and $z[r]$ the number of root nodes $s\in S$ in subtree $T_r|u$ is set to $1$ if $u\in S$, and to $0$ otherwise. In Line \ref{line:gmrctinitmsg}, the message $msg_p[r]$ is initialized in case node $u$ is a leaf in tree $T_r$, and therefore is sent to the parent of $u$ in $T_r$ before $\mathcal{A}_{MRCT}'$ is executed the first time. If $u$ is not a leaf in in tree $T_r$, then the values of $rc_S[r]$, $z[r]$ and $msg_p[r]$ will change during the execution of Algorithm\iffull~\ref{alg:TreeAggregating}\fi~\textsc{TreeAggregating} .

Back in Algorithm\iffull~\ref{alg:MRCT}\fi~\textsc{MRCT}  in Lines \ref{line:MRCTinftystart} to \ref{line:MRCTinftyend}, after the execution of Algorithms\iffull~\ref{alg:TreeComputing}\fi~\textsc{TreeComputing}  and\iffull~\ref{alg:TreeAggregating}\fi~\textsc{TreeAggregating} , all nodes $r\in S$ store the routing cost $rc_S[r]$ of their tree $T_r$ in $RC_S(r)$, all other nodes $v\notin S$ store $\infty$. Then, in Line \ref{line:MRCTminstart} to \ref{line:MRCTminend} the tree with the smallest routing cost is selected and broadcast.

\begin{proof}[Proof of Theorem \ref{theo:mrctapprox}]
Note that our Algorithms compute exactly the same trees as Algorithm 1 in \cite{hochuli:holzer:MRCST}, as received messages (which determine the parent) from neighbors are forwarded based on the same rule (prioritizing neighbor IDs) in both algorithms.
As a consequence Algorithm\iffull~\ref{alg:MRCT}\fi~\textsc{MRCT}  computes the same results as the S-MRCT algorithm in \cite{hochuli:holzer:MRCST} by aggregating along the trees computed by Algorithm\iffull~\ref{alg:TreeComputing}\fi~\textsc{TreeComputing}. Lines \ref{line:MRCTstart} to \ref{line:MRCTendestim}, \ref{line:MRCTaggmin} and \ref{line:MRCTbroadmin} take time $\BO(D_\omega)$ to estimate, broadcast and aggregate along tree $T_1$. Algorithms\iffull~\ref{alg:TreeComputing}\fi~\textsc{TreeComputing}  and\iffull~\ref{alg:TreeAggregating}\fi~\textsc{TreeAggregating}  are bound by time $\BO(|S|+D_\omega)$, as seen before. Hence, Algorithm\iffull~\ref{alg:MRCT}\fi~\textsc{MRCT}  runs in time $\BO(|S|+D_\omega)$. Theorem 2 in \cite{hochuli:holzer:MRCST} states that the algorithm that we transfer into our setting computes a $(\times,2)$-approximation to an MRCT. Note that \cite{hochuli:holzer:MRCST} assumes that $\omega(e)$ is uniform or corresponds to the time needed to transmit a message through edge $e$.
\end{proof}

\begin{proof}[Proof of Theorem \ref{theo:mrctapprox-rand}]
Algorithm \textsc{MRCTrand} is very similar to Algorithm~\ref{alg:MRCT} \textsc{MRCT}, only the set of root nodes is different. 
Like in \cite{hochuli:holzer:MRCST} we sample a set $S'\subset S$ of size $\BO(\log n)$ uniformly at random and then execute Algorithm\iffull~\ref{alg:MRCT}\fi~\textsc{MRCT} on set $S'$. In the end we choose the tree with smallest routing cost as the approximation. Theorem 3 in \cite{hochuli:holzer:MRCST} states this yields a $(\times,2+\epsilon)$-approximation to an MRCT. From the proof of Theorem \ref{theo:mrctapprox} we know that the runtime of opur implementation is $\BO(|S'|+D_\omega)=\BO(D_\omega + \log n)$. Note that \cite{hochuli:holzer:MRCST} assumes that $\omega(e)$ is uniform or corresponds to the time needed to transmit a message through edge $e$.
\end{proof}

\begin{algorithm}[htb]
\begin{algorithmic}[1]
\IF{$u=1$}\label{line:MRCTstart}
        \STATE \textbf{estimate} $D_{\omega}'$ by generating a spanning tree $T_1$ and using Fact \ref{fact:ecc-approx-diam};
        \STATE \textbf{broadcast} $D_\omega'$ through $T_1$;
\ELSE
        \STATE \textbf{receive} value of $D_\omega'$;
\ENDIF \label{line:MRCTendestim}
\STATE \textbf{execute} Algorithm\iffull~\ref{alg:TreeComputing}\fi~\textsc{TreeComputing}  with $(|S|,D_\omega', \bot)$; \label{line:MRCTalgdist}
\STATE \textbf{execute} Algorithm\iffull~\ref{alg:TreeAggregating}\fi~\textsc{TreeAggregating}  with $(|S|, D_\omega',\mathcal{A}_{MRCT}'())$; \label{line:MRCTalgagg}
\IF{$u\in S$}\label{line:MRCTinftystart}
        \STATE $RC_S(T_u)=rc_S[u]$;
\ELSE
        \STATE $RC_S(T_u)=\infty$;
\ENDIF \label{line:MRCTinftyend}
\IF{$u=1$} \label{line:MRCTminstart}
        \STATE \textbf{aggregate} $v:= \arg\min_{v\in V}RC_S(T_v)$ via $T_1$; \label{line:MRCTaggmin}
        \STATE \textbf{broadcast} ``$T_v$ is $(\times,2)$-approximation for S-MRCT''; \label{line:MRCTbroadmin}
\ENDIF \label{line:MRCTminend}
\end{algorithmic}
\caption{\textit{MRCT} Approximation (executed by node $u$)}\label{alg:MRCT}
\vspace*{0.5cm}
\end{algorithm}

\begin{algorithm}[htb]
\begin{algorithmic}[1]
\STATE \COMMENT{INITIALIZATION} 
\STATE \textbf{global} $rc_{S}:=\{0,\dots,0\}$; \label{line:gmrctinitstart}
\IF{$u\in S$}
        \STATE $z:=\{1,\dots,1\}$;
\ELSE
        \STATE $z:=\{0,\dots,0\}$; 
\ENDIF \label{line:gmrctinitend}
\STATE \textbf{for each} $v\in S$ \textbf{do} $msg_p[v]=(rc_S[v], z[v])$; \COMMENT{output if $u$ is a leaf}\label{line:gmrctinitmsg}\newline

\STATE \COMMENT{COMPUTATION, $g(v, msg_c, u_c)$: $v$ root node ID of the message, $msg_c$ message of a child $u_c$ in $T_v$}
\STATE $(r_c, z_c):=msg_c$;
\STATE $rc_S[v]:=rc_S[v]+r_c+\omega(u,u_c)\cdot z_c\cdot(|S|-z_c)$;
\STATE $z[v]:=z[v]+z_c$;
\STATE $msg_p[v]=(rc_S[v], z[v])$; \COMMENT{output}
\end{algorithmic}
\caption{$\mathcal{A}_{MRCT}'()$}
\label{alg:MRCT_g}
\vspace*{0.5cm}
\end{algorithm}

\fi


\addcontentsline{toc}{section}{References} 
\bibliographystyle{abbrv}
\bibliography{references}

\begin{center}
\textbf{Appendix}
\end{center}


\subsection{Tree Variation}\label{app:tree}

In schedule $\tau$ initialized in Algorithm \ref{alg:DLGcomp} (Line \ref{line:tauinit}) we only store the arrival time of the first $r$-message per root node $r\in S$.
Therefore the second dimension of the array $\tau$ representing the schedule can be omitted in both, Algorithm \ref{alg:DLGcomp} and \ref{alg:DLGagr}.
The number of parents $sp$ can be omitted too.
And since we have to handle only one $r$-message per root node $r$, we rewrite Lines \ref{line:treevarstart}--\ref{line:treevarend} of Algorithm \ref{alg:DLGcomp} as stated in Algorithm \ref{alg:TreeComputing}.

In Algorithm\iffull~\ref{alg:DLGagr}\fi~\textsc{DLGagr} just a minor modification is required. 
The \textit{`For each'}-loop in Line \ref{line:sendreceiveagg} needs to iterate over one parent $parent\_in\_T_r$ per tree $T_r$ only, since in a tree a node has at most one parent.

The fully rewritten Tree Variation of Algorithms\iffull~\ref{alg:DLGcomp}\fi~\textsc{DLGcomp} and \iffull~\ref{alg:DLGagr}\fi~\textsc{DLGagr} are stated in Algorithm \ref{alg:TreeComputing} and \ref{alg:TreeAggregating}.

\begin{proof}[Proof of Theorem \ref{theo:tree}]
By modifying Algorithms\iffull~\ref{alg:DLGcomp}\fi~\textsc{DLGcomp} and \iffull~\ref{alg:DLGagr}\fi~\textsc{DLGagr}, we only changed the space requirements inside the nodes $u\in V$ and did not change anything concerning the running time. Therefore the time complexity remains $\BO(|S|+k)$  after the modifications.
To show that the modification computes valid trees, note that a tree $T_r$ is a subgraph of a DLG $\calL_r$. A tree $T_r$ can be computed from a DLG $\calL_r$, by removing all but one edge to the parents at each node $u\in V$, except for the root node $r$, which has no parent node at all. This is done in Algorithm \ref{alg:TreeComputing} by considering only the neighbor node $u_p$ as a parent in tree $T_r$, which sends the first $r$-message.
\end{proof}

\begin{algorithm}[H]
\small
\begin{algorithmic}[1]
\STATE \COMMENT{On node globally usable variables:}
\STATE $L:=\emptyset$, $L_{delay}:=\emptyset$;
\STATE $\tau:=[\infty,\infty,\dots,\infty]$; \COMMENT{$\tau[v]=$ time when the first message of tree $T_v$ reaches $u$}
\STATE $\omega:=\{\infty,\infty,\dots,\infty\}$; \COMMENT{$\omega[v]=\omega(v,u)$, $\omega[\emptyset]:=\infty$}
\IF{$u \in S$}
        \STATE $L:=\{u\}$;
        \STATE $\omega[u]:=0$;
        \STATE $\tau[u]:=0$;
\ENDIF
\STATE $L_1,\dots,L_{deg(u)}:=L$;
\STATE \textbf{initialize} algorithm $\mathcal{A}_{agr}$;
\FOR{$t=0,\dots,|S|+D_{\omega}'$}
        \FOR{$i=1,\dots,deg(u)$}
                \STATE $l_i:=\min\left\{v\in L_i : \tau[v]+\omega(u,v)\ge t\wedge v\notin L_{delay}\right\}$;
        \ENDFOR
        \STATE within one time slot:
        \newline \textbf{send} $\left(l_1,\omega[l_1]+\omega\left(u,u_1\right), msg_c[l_1]\right)$ to neighbor $u_1$, receive $\left(r_1,\omega_1, msg_{p,1}\right)$ from $u_1$;
        \newline \textbf{send} $\left(l_2,\omega[l_2]+\omega\left(u,u_2\right), msg_c[l_2]\right)$ to neighbor $u_2$, receive $\left(r_2,\omega_2, msg_{p,2}\right)$ from $u_2$;
        \newline $\dots$
        \newline \textbf{send} $\left(l_{deg(u)},\omega[l_{deg(u)}]+\omega\left(u,l_{deg(u)}\right), msg_c[l_{deg(u)}]\right)$ to neighbor $u_{deg(u)}$, receive $\left(r_{deg(u)},\omega_{deg(u)}, msg_{p,deg(u)}\right)$ from $u_{deg(u)}$;
        
        \STATE $R=\{r_i:r_i<l_i\text{ and }i\in \{1,\dots,deg(u)\}\}\setminus L$;
        \STATE $s:=\left\{
          \begin{array}{l l}
            \infty & \quad \text{if $L_{delay} = \emptyset$}\\
            \min(L_{delay}) & \quad \text{else}
          \end{array} \right.$;
        \IF{$s\leq \min(R) \textbf{ and } s<\infty$}
                \STATE $L_{delay}:=L_{delay}\setminus\{s\}$;
        \ENDIF
        
        \FOR{$i=1,\dots,deg(u)$}
                \IF[$T_{l_i}$'s message will be delayed due to $T_{r_i}$.]{$r_i<l_i$}
                        \IF{$r_i\notin L$}
                                \STATE $\tau[r_i]:=t$;
                                \STATE $parent\_in\_T_{r_i}:=$ neighbor $i$;
                                \STATE \textbf{execute} algorithm $\mathcal{A}_{agr}(r_i, msg_{p,i})$;
                                \STATE $\omega[r_i]=\omega_i$;
                                \STATE $L:=L\cup\{r_i\}, L_1:=L_1\cup\{r_i\}, L_2:=L_2\cup\{r_i\},\dots,L_{i-1}:=L_{i-1}\cup\{r_i\},L_{i+1}:=L_{i+1}\cup\{r_i\},\dots,L_{deg(u)}:=L_{deg(u)}\cup\{r_i\}$;
                                \IF{$\min(R)<r_i \textbf{ or } s<r_i$}
                                        \STATE $L_{delay}:=L_{delay}\cup\{r_i\}$;
                                \ENDIF
                        \ELSE
                                \STATE $L_i:=L_i\setminus \{r_i\}$;
                        \ENDIF
                \ELSE
                        \STATE $L_i:=L_i\setminus \{l_i\}$; \COMMENT{$T_{l_i}$'s message was successfully sent to neighbor $i$.}
                \ENDIF
        \ENDFOR
\ENDFOR
\end{algorithmic}
\caption{computing $|S|$ trees (executed by node $u$) \newline 
        \textbf{Input:} $|S|$, $D_\omega'$, $\mathcal{A}_{agr}$ \newline 
        \textbf{passed parameters on execution of $\mathcal{A}_{agr}(v, msg_p[v])$:} $msg_p[v]$: message of the parent in tree $T_v$, $msg_c[v]$: message that is sent to children in tree $T_v$ after execution}
\label{alg:TreeComputing}
\end{algorithm}

\begin{algorithm}[ht]
\begin{algorithmic}[1]
\STATE \textbf{initialize} algorithm $\mathcal{A}_{agr}'$;
\FOR{$t=0,\dots,|S|+D_{\omega}'$}
        \STATE within one time slot:
        \newline \textbf{For each} $v\in L$ \\
                 \ \ \ \textbf{if} $t=|S|+D_{\omega}'-\tau[v]$ \textbf{then}\\
                         \ \ \ \ \ \ \ \textbf{send} $\left(v,msg_p[v]\right)$ to $parent\_in\_T_{v}$;
        \newline \textbf{receive} $\left(v_1,msg_{c,u_1}\right)$ from neighbor $u_1$;
        \newline \textbf{receive} $\left(v_2,msg_{c,u_2}\right)$ from neighbor $u_2$;
        \newline $\cdots$
        \newline \textbf{receive} $\left(v_{deg(u)},msg_{c,u_{deg(u)}}\right)$ from neighbor $u_{deg(u)}$;
        \FOR{$i=1,\dots,deg(u)$}
                \IF{$v_i\neq\bot$}
                        \STATE \textbf{execute} algorithm $\mathcal{A}_{agr}'(v_i, msg_{c,u_i})$;
                \ENDIF
        \ENDFOR
\ENDFOR
\end{algorithmic}
\caption{Tree aggregation (executed by node $u$) \newline
        \textbf{Input:} $|S|$, $D_\omega'$, $\mathcal{A}_{agr}'$   \newline
        \textbf{passed parameters on execution of $\mathcal{A}_{agr}'(v,msg_c)$:} $msg_c$: message of a child in tree $T_v$, $msg_p[v]$: message that is sent to parents in tree $T_v$ after execution}
\label{alg:TreeAggregating}
\end{algorithm}

\end{document}